\documentclass[10pt, twocolumn]{IEEEtran}
\usepackage{subfigure, epsfig, color}
\usepackage{amsmath}
\usepackage{amsthm}
\usepackage{amssymb}
\usepackage{multirow}
\usepackage{booktabs}

\newtheorem{theorem}{Theorem}
\newtheorem{lemma}{Lemma}

\usepackage{mathtools}
\usepackage{algorithm,algorithmic}
\usepackage{cite}
\usepackage{xr}

\begin{document}
%
\title{Wideband Channel Estimation with A Generative Adversarial Network} 

\author{Eren Balevi and
        Jeffrey G. Andrews
\thanks{The authors are with the Wireless Networking and Communications Group, Dept. of Electrical and Comp. Eng., The University of Texas at Austin, TX, 78712, USA. Email:
erenbalevi@utexas.edu and jandrews@ece.utexas.edu.} 
}

\maketitle
\vspace{-0.7in}

\begin{abstract}
Communication at high carrier frequencies such as millimeter wave (mmWave) and terahertz (THz) requires channel estimation for very large bandwidths at low SNR. Hence, allocating an orthogonal pilot tone for each coherence bandwidth leads to excessive number of pilots. We leverage generative adversarial networks (GANs) to accurately estimate frequency selective channels with few pilots at low SNR. The proposed estimator first learns to produce channel samples from the true but unknown channel distribution via training the generative network, and then uses this trained network as a prior to estimate the current channel by optimizing the network's input vector in light of the current received signal.  Our results show that at an SNR of $-5$ dB, even if a transceiver with one-bit phase shifters is employed, our design achieves the same channel estimation error as an LS estimator with SNR $= 20$ dB or the LMMSE estimator at $2.5$ dB, both with fully digital architectures. Additionally, the GAN-based estimator reduces the required number of pilots by about $70\%$ without significantly increasing the estimation error and required SNR. We also show that the generative network does not appear to require retraining even if the number of clusters and rays change considerably.
\end{abstract}


\section{Introduction}
Millimeter wave (mmWave) and terahertz (THz) communication offer large untapped bandwidths. Low signal-to-noise-ratio (SNR) ($\ll 0$ dB) channel estimation at these frequencies and bandwidths is desirable before beam alignment is completed, because exhaustively searching all narrow beams without estimating the channel brings exponentially increasing sample and computational complexity. However, there are nontrivial challenges regarding channel estimation  stemming from the propagation physics, the necessity of using a great many antenna elements for high gain beamforming, and the hybrid transceiver architectures.

The contribution of this paper is to adapt powerful deep generative models for frequency selective mmWave and THz channel estimation as an alternative to known sparsity-based compressed sensing algorithms, which have been used for narrowband channels \cite{Alkhateeb2014Heath,Alkhateeb2015Heath,Ria2016Heath}, frequency-selective channels \cite{Schniter2014Sayeed,GaoHuDai2016Wang,Venugopal17Heath} and low-resolution channels \cite{Mo2018Heath}. Deep generative models offer an appealing approach to exploiting sparsity, as they can use knowledge of a finite number of signals from a class to learn a basis for the whole class. Furthermore, these models enable us to solve optimization problems with a simple and fast gradient descent based method. As an additional benefit, generative models can exploit the overall cross-correlations among the frequency, time and spatial domains, which have been traditionally ignored to simplify the estimator \cite{Li02}. Among the two prominent deep generative models, namely variational autoencoders (VAEs) \cite{kingma2013auto} and generative adversarial networks (GANs) \cite{goodfellow2014generative}, we utilize a GAN in this paper, since GANs can very effectively compress the signals to a low dimensional manifold by leveraging the channel structures. Exploiting this property enables us to reduce the number of required pilots for accurate channel estimation. 

\subsection{Motivation and Related Work}
High frequency bands incur high propagation losses for terrestrial communication, and hence a large number of small antenna elements is needed to attain a large beamforming gain. Conventional estimators require that the number of pilots has to be at least equal to the number of transmit antennas to avoid having an ill-posed problem. Thus, to reduce the number of pilots the existing high bandwidth channel estimators have been centered around compressed sensing tools motivated by the sparsity of mmWave channels \cite{Rap2013ItWillWork}. The same approach was also used for sub-6 GHz massive MIMO channels \cite{Gao2016Wang, Nguyen2013Ghrayeb, Gao2015Chen, Lin2017Jiang}. Unfortunately, it is very hard (or impossible) to find the basis that would yield the sparsest representation.  Also, the reconstruction phase is complex and slow for compressed sensing algorithms, which require the solution of an optimization problem to find the locations and values of the sparse coefficients.  This restricts their usage to channels with fairly long coherence intervals.    

Deep learning has been recently utilized as an alternative for high-dimensional channel estimation. Specifically, \cite{Soltani18Sheikhzadeh} uses convolutional neural networks (CNNs) to make interpolation and denoising in 2-dimensional OFDM channels, \cite{He18Li} incorporates a special CNN as a denoiser to the approximate message passing (AMP) algorithm for beamspace channels, and \cite{Dong19Gaspar} adapts CNNs for 3-dimensional channels to exploit the correlations in frequency, time and space. To prevent the training complexity of CNNs for channel estimation, \cite{BalAndDCE,BalDosAndMassiveMIMO} proposed an untrained neural network that can precede or follow a least-squares (LS) estimator for OFDM and MIMO-OFDM channels, respectively. Combining an LS estimator with neural networks was also proposed in \cite{yang2019deep}, \cite{ru2019model}. There are some other studies that consider deep learning to tackle the detrimental effects of quantization for channel estimation \cite{GaoLi18}, \cite{BalAnd19}.

What distinguishes this paper from the prior techniques is that we design a GAN to learn to produce channel samples according to its distribution, and then use this knowledge as a priori information to estimate the actual current channel. This is a quite different approach from using GANs for channel modeling \cite{Shea2019West}, \cite{Ye2018Sivanesan}. Furthermore, as opposed to AMP-based channel estimators, our GAN approach does not require us to know or model the channel distribution. Instead, the GAN learns to produce samples that statistically are very close to the true but unknown channel distribution. The closest papers to our work are our recent papers \cite{Dos20BalAnd}, \cite{Bal2020DosJalDimAnd}, which use a similar GAN-based channel estimation architecture to reduce the number of pilots for single stream narrowband (frequency flat) massive MIMO channel estimation, with very tight ($\lambda/10$) antenna spacing contributing extra spatial correlation.  In contrast, in the current paper we introduce a frequency selective channel and utilize a more standard planar array with $\lambda/2$ spaced antennas.  Additionally, we consider one-bit quantized phase shifters to further decrease the power consumption and hardware costs for such a large antenna array.  We also study the generalization capability both through analysis and experiments. 


\subsection{Contributions}
The main contribution of this paper is to propose and study a novel GAN-based channel estimation algorithm for wideband frequency selective channels. Although in this paper the modulation and demodulation is based on OFDM, the proposed approach can be adapted to single-carrier frequency domain equalization (SC-FDE) systems as long as the channel is estimated in frequency domain. We will demonstrate that our GAN-based framework can estimate channels at very low SNR with a reduced number of pilots for hybrid beamforming architectures when the channel estimation is formulated as an inverse problem. In addition to the novel architecture, our contributions are both theoretical and empirical.

Theoretically, there are two main contributions.  First, the GAN framework requires sub-Gaussian measurements to meet theoretical guarantees \cite{bora2017compressed}.  In the channel estimation case, these measurements are determined by the pilots and the digital and analog precoders/combiners. We prove that when the pilots are chosen as zero mean bounded i.i.d. random variables, the sub-Gaussian requirement is indeed met for channel estimation even if there are constraints due to phase shifters and total transmission power.  Thus the corresponding guarantees hold.  Second, we investigate the generalization capability of the proposed estimator for channels with a different number of clusters and rays than the channel used for training the GAN.  Our technical approach is to apply theoretical principles from reinforcement learning.

Our empirical results demonstrate that the major challenges -- hybrid transceivers, low SNRs and insufficient pilots -- can be tackled with our proposed estimator. Specifically, we benchmark our technique versus the performance of conventional channel estimators for fully digital transceivers.   We find that our technique at an ultra-low SNR of $-5$ dB matches the performance of LS estimation at 20 dB and linear minimum mean square error (LMMSE) estimation at 2.5 dB. Furthermore, it is shown that GANs provide a lower channel estimation error than the traditional CNNs that are not trained with adversary loss, e.g., ResNet due to exploiting the high channel correlations much more efficiently. Additionally, our estimator allows a significant reduction in pilot tones (more than $50\%$) without any substantial performance loss, and yields lower estimation error than the Orthogonal Matching Pursuit (OMP) algorithm in this regime. 


\section{System Model and Problem Statement}
We consider single user communication to estimate the frequency selective channel over a large number of antennas via pilot symbols. However, all the ideas proposed in this paper are equivalently applicable to multi-user communication if orthogonal pilots are allocated to each user and there is no inter-beam interference\footnote{Note that the proposed estimator is robust to inter-beam interference as long as the interference has a distribution whose tails are exponentially bounded as has been recently proven in \cite{Jalal2020compressed}. On the other hand, for heavy-tailed interference novel reconstruction methods are needed instead of minimizing the Euclidean distance.}. In the case of large antenna arrays, having a dedicated RF chain per antenna is too costly in terms of hardware and power consumption. Thus, the number of RF chains is reduced by processing the signals both in the digital and analog domain. This architecture is illustrated in Fig. \ref{fig:Hybrid_Beamforming}. Here, $N_s$ data streams are precoded digitally at each subcarrier. Then, the precoded signal is OFDM modulated for the $N_t^{\rm RF}$ RF chains, processed with an analog precoder (or phase shifters) and  transmitted over the $N_t$ transmit antennas. Similarly, the receiver has an analog combiner that converts the $N_r$ dimensional received signal into an $N_r^{\rm RF}\times 1$ vector. The resultant signal is then OFDM demodulated and combined with the digital combiner at each subcarrier. 
\begin{figure*} [!t]
\centering
\includegraphics[width=7in]{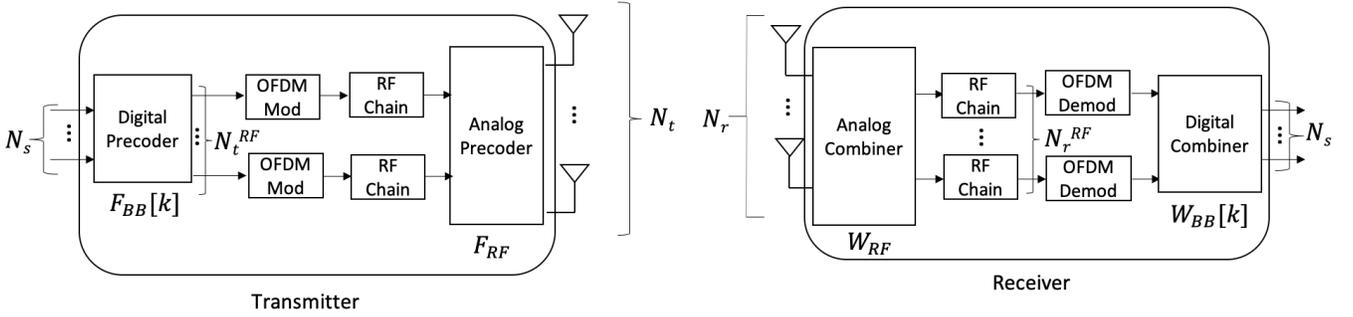}
\caption{The communication system model utilized for channel estimation, in which the pilots are passed through digital and analog precoders and combiners.}
\label{fig:Hybrid_Beamforming}
\end{figure*}

\subsection{Hybrid Transceivers}\label{System Model}
The pilot tone $\mathbf{p}[n,k]$ with the first and second axes being the time and frequency index is an $N_s\times 1$ vector such that $\mathbb{E}[\mathbf{p}[n,k]\mathbf{p}[n,k]^H] = \frac{1}{N_s}\mathbf{I}_{N_s}$. This signal is processed with the $N_t^{\rm RF} \times N_s$ dimensional digital precoder matrix $\mathbf{F}_{\rm BB}[n,k]$ as
\begin{equation}\label{dig_prec_sig}
    \mathbf{s}[n,k] = \mathbf{F}_{\rm BB}[n,k]\mathbf{p}[n,k]
\end{equation}
where $\mathbf{s}[n,k] = [s_1[n,k], s_2[n,k], \cdots, s_{N_t^{\rm RF}}[n,k]]^T$ for $k=0,1,\cdots,N_f-1$, i.e., there are $N_f$ subcarriers and $s_q[n,k]$ corresponds to the pilot for the $q^{th}$ RF chain on the $k^{th}$ subcarrier. In accordance with \eqref{dig_prec_sig}, the time domain samples become
\begin{equation}
    \mathbf{u}[n] = (\mathbf{F}^H\otimes\mathbf{I}_{\rm N_t^{\rm RF}})\underbrace{\begin{pmatrix}\mathbf{s}[n,0] \\ \vdots \\ \mathbf{s}[n,N_f-1]\end{pmatrix}}_{\mathbf{s}[n]}
\end{equation}
where $\mathbf{F}^H$ is the $N_f\times N_f$ IDFT matrix and $\otimes$ denotes the Kronecker product. Transmitting the pilots after the $ N_t \times N_t^{\rm RF}$ analog precoder $\mathbf{F}_{\rm RF}$\footnote{Since it is not practical to change the analog precoder and combiner at each symbol time, we omit the index $n$.} over a frequency-selective channel yields the received signal
\begin{equation}
    \mathbf{r}[n] = \sqrt{\rho}\mathbf{C}(\mathbf{I}_{\rm N_f}\otimes\mathbf{F}_{\rm RF})\mathbf{u}[n] + \mathbf{v}[n]
\end{equation}
where  $\rho$ is the average received power,
\begin{equation} \label{chan_mat}
    \mathbf{C} = 
    \begin{pmatrix}
    \mathbf{C}_0 & \mathbf{C}_{N_f-1} & \mathbf{C}_{N_f-2} & \cdots & \mathbf{C}_1 \\ 
    \mathbf{C}_1 & \mathbf{C}_0 & \mathbf{C}_{N_f-1} & \cdots & \mathbf{C}_2 \\ 
    \mathbf{C}_2 & \mathbf{C}_1 & \mathbf{C}_{0} & \cdots & \mathbf{C}_3 \\ 
    \vdots & \vdots & \vdots & \ddots & \vdots  \\
    \mathbf{C}_{N_f-1} & \mathbf{C}_{N_f-2} & \mathbf{C}_{N_f-3} & \cdots & \mathbf{C}_0 
    \end{pmatrix}
\end{equation}
where $\mathbf{C}_l$ is an $N_r\times N_t$ channel matrix. We assume that there are $L_c$ channel taps in the time domain. Hence, the matrix $\mathbf{C}_l$ is non-zero only for $l=0,1,\cdots, L_c-1$ and it is a zero matrix for $l=L_c, L_c+1,\cdots, N_f-1$. The noise vector $\mathbf{v}[n]$ is a zero mean i.i.d. Gaussian with covariance matrix $\sigma_n^2\mathbf{I}_{\rm N_fN_r}$. 

Combining the received signal with the phase shifters, leading to the $ N_r \times N_r^{\rm RF}$ matrix $\mathbf{W}_{\rm RF}$ followed by the DFT matrix $\mathbf{F}$ for OFDM demodulation gives
\begin{align}
    \mathbf{y}[n] = &\sqrt{\rho}(\mathbf{F}\otimes\mathbf{I}_{\rm N_r^{\rm RF}})(\mathbf{I}_{\rm N_f}\otimes\mathbf{W}_{\rm RF}^H)\mathbf{C}(\mathbf{I}_{\rm N_f}\otimes\mathbf{F}_{\rm RF}) 
    \nonumber \\ &
    (\mathbf{F}^H\otimes\mathbf{I}_{\rm N_t^{\rm RF}})\mathbf{s}[n] + \mathbf{w}[n]\label{analog_comb_sig}
\end{align}
where $\mathbf{w}[n]=(\mathbf{F}\otimes\mathbf{I}_{\rm N_r^{\rm RF}})(\mathbf{I}_{\rm N_f}\otimes\mathbf{W}_{\rm RF}^H)\mathbf{v}[n]$. Due to the mixed-product property of the Kronecker product, i.e., $(\mathbf{V}\otimes \mathbf{Y})(\mathbf{X}\otimes \mathbf{Z}) = (\mathbf{V}\mathbf{X})\otimes(\mathbf{Y}\mathbf{Z})$, it follows that
\begin{equation}\label{mix_prod_kron}
    (\mathbf{F}\otimes\mathbf{I}_{\rm N_r^{\rm RF}})(\mathbf{I}_{\rm N_f}\otimes\mathbf{W}_{\rm RF}^H) = (\mathbf{I}_{\rm N_f}\otimes\mathbf{W}_{\rm RF}^H)(\mathbf{F}\otimes\mathbf{I}_{\rm N_r}).
\end{equation}
Substituting \eqref{mix_prod_kron} into \eqref{analog_comb_sig} and doing the same for the transmitter lead to
\begin{align}
    \mathbf{y}[n] = &\sqrt{\rho}(\mathbf{I}_{\rm N_f}\otimes\mathbf{W}_{\rm RF}^H)(\mathbf{F}\otimes\mathbf{I}_{\rm N_r})\mathbf{C}(\mathbf{F}^H\otimes\mathbf{I}_{\rm N_t}) \nonumber \\ &
    (\mathbf{I}_{\rm N_f}\otimes\mathbf{F}_{\rm RF}) \mathbf{s}[n] + \mathbf{w}[n]. \label{rec_sig}
\end{align}
Since $\mathbf{C}$ is a block-diagonal circulant matrix as given in \eqref{chan_mat}, it is diagonalized with its left and right multiplying terms in \eqref{rec_sig} to
\begin{equation} 
    \mathbf{H} = (\mathbf{F}\otimes\mathbf{I}_{\rm N_r})\mathbf{C}(\mathbf{F}^H\otimes\mathbf{I}_{\rm N_t}) =
    \begin{pmatrix}
    \mathbf{H}_0 & 0 & \cdots & 0 \\ 
    0 & \mathbf{H}_1 & \cdots & 0 \\ 
    \vdots  & & \ddots & \vdots \\ 
     0 & \cdots & 0 & \mathbf{H}_{N_f-1} 
    \end{pmatrix}
\end{equation}
such that
\begin{equation} \label{chan_at_one_subcar}
    \mathbf{H}_k = \sum_{l=0}^{L_c-1}\mathbf{C}_le^{-j2\pi kl/N_f}
\end{equation}
for $k=0,1,\cdots,N_f-1$. Hence,
\begin{equation} \label{diag_rec_sig}
    \mathbf{y}[n] = \sqrt{\rho}(\mathbf{I}_{\rm N_f}\otimes\mathbf{W}_{\rm RF}^H)\mathbf{H}(\mathbf{I}_{\rm N_f}\otimes\mathbf{F}_{\rm RF}) \mathbf{s}[n] + \mathbf{w}[n]. 
\end{equation}
 
Vectorizing the channel matrix $\mathbf{H}$ in \eqref{diag_rec_sig} yields
\begin{equation} \label{vector_rec_sig_per_symb}
\mathbf{y}[n] = \sqrt{\rho} \underbrace{(\mathbf{s}[n]^T(\mathbf{I}_{\rm N_f}\otimes\mathbf{F}_{\rm RF}^T)\otimes(\mathbf{I}_{\rm N_f}\otimes\mathbf{W}_{\rm RF}^H))}_{\mathbf{A}[n]}\mathbf{\underbar{h}} + \mathbf{w}[n] 
\end{equation}
where $\mathbf{\underbar{h}} = {\rm vec}(\mathbf{H})$. We assume that there are $N_p$ symbols in one frame and the channel is constant throughout this frame. Thus,
\begin{equation}\label{chn_est_prob}
    \mathbf{y} = \begin{pmatrix}\mathbf{y}[n] \\ \vdots \\ \mathbf{y}[n+N_p-1]\end{pmatrix} =  \sqrt{\rho}\underbrace{(\mathbf{I}_{\rm N_p}\otimes\mathbf{A}[n])}_{\mathbf{A}}\underbrace{\begin{pmatrix}\mathbf{\underbar{h}} \\ \vdots \\ \mathbf{\underbar{h}}\end{pmatrix}}_{\mathbf{h}} + \mathbf{w}
\end{equation}
where $\mathbf{w}$ is obtained by concatenating the time domain samples as was done for $\mathbf{y}$. Here, $\mathbf{A}$ refers to the measurement matrix, and $\mathbf{h}$ (or $\mathbf{\underbar{h}}$) is the target signal that we aim to estimate. 

Optimizing the beamformer and combiner matrices in \eqref{chn_est_prob} enhances the received SNR, but unfortunately there is no way to optimally set these without knowing the channel. Hence, we consider an arbitrary scenario such that the digital beamformer and combiner are set to the identity matrix, i.e., $\mathbf{F}_{\rm BB}[n,k] = \mathbf{I}_{\rm N_t^{RF} \times N_s}$ and $\mathbf{W}_{\rm BB}[n,k] = \mathbf{I}_{\rm N_r^{RF} \times N_s}$. Furthermore, the phase shifters are adjusted with one-bit quantized angles to further reduce the power consumption of transceivers. This means that $[\mathbf{F}_{\rm RF}]_{i,j} = \frac{1}{\sqrt{N_t}}e^{j\theta_{i,j}}$, and $[\mathbf{W}_{\rm RF}]_{i,j} = \frac{1}{\sqrt{N_r}}e^{j\phi_{i,j}}$, in which $\theta_{i,j}, \phi_{i,j} \in A$, where $A = \{0, \pi \}$, and $[\mathbf{F}_{\rm RF}]_{i,j}$ and $[\mathbf{W}_{\rm RF}]_{i,j}$ are the $(i,j)$th element of $\textbf{F}_{\rm RF}$ and $\textbf{W}_{\rm RF}$, respectively.

\subsection{OFDM Channel Estimation with Multiple Antennas}
The optimum channel estimator for \eqref{chn_est_prob} is found via the \textit{maximum a posteriori} (MAP) optimization, which is equivalent to
\begin{eqnarray}\label{MAP}
\mathbf{\hat{h}_{\rm MAP}} & = & \underset{\mathbf{h}}{\text{arg\ max\ }}\mathcal{\log(P(\mathbf{y}|\mathbf{h}))}+\mathcal{\log(P(\mathbf{h}))}.
\end{eqnarray}
The main challenges for MAP estimation are the prohibitive computational complexity, since the coefficients of $\mathbf{h}$ are mixed in $\mathbf{y}$ and this makes the calculation of $\mathcal{P(\mathbf{y}|\mathbf{h})}$ quite complex, and the need for channel distribution. As a special case, when $\mathcal{P(\mathbf{y}|\mathbf{h})}$ is Gaussian, $\mathbf{\hat{h}_{\rm MAP}}$ becomes equivalent to the LMMSE estimator 
\begin{equation}\label{LMMSE}
    \mathbf{\hat{h}_{\rm LMMSE}} = \mathbf{R}_h (\mathbf{R}_h+\rho^{-1}\mathbf{\Gamma})^{-1}\mathbf{\hat{h}_{\rm LS}}
\end{equation}
where 
\begin{equation}\nonumber
    \mathbf{R}_h=\mathbb{E}[\mathbf{h}\mathbf{h}^H],
\end{equation}
\begin{equation}\nonumber
    \mathbf{\Gamma} = (\mathbf{A}^H\mathbf{A})^{-1}\mathbf{A}^H\mathbb{E}[\mathbf{w}\mathbf{w}^H]\mathbf{A}(\mathbf{A}^H\mathbf{A})^{-1},
\end{equation}
and
\begin{equation}\label{LS}
    \mathbf{\hat{h}_{\rm LS}} = \frac{1}{\sqrt{\rho}}(\mathbf{A}^H\mathbf{A})^{-1}\mathbf{A}^H\mathbf{y}.
\end{equation}
Note that $\mathbf{\Gamma}$ becomes an identity matrix in the case of digital transceivers. However, \eqref{LMMSE} is still computationally expensive due to the matrix inversions.  Also, $\mathbf{A}^H\mathbf{A}$ becomes non-invertible if there are not sufficient pilots. The AMP algorithm can near optimally solve \eqref{MAP} with reasonable complexity if $\mathcal{P}(\mathbf{h})$ is known \cite{Donoho2009Montanari}. However, it is unrealistic to assume a known $\mathcal{P}(\mathbf{h})$. Modeling $\mathbf{h}$ with Gaussian mixtures whose parameters are found with Expectation-Maximization algorithm can be a method if the elements of $\mathbf{h}$ are independent \cite{vila2013expectation}. However, the entries of $\mathbf{h}$ are correlated in wireless channels. Finding a sparsifying basis for the channel in \eqref{chn_est_prob} and using OMP and Basis Pursuit Denoising (or Lasso) for channel estimation lead to a high performance loss \cite{Bal2020DosJalDimAnd}.

For multiple antenna OFDM channel estimation, we use a fundamentally different approach. Our key idea is to design a GAN that learns to produce plausible channel samples instead of finding or modeling the highly complex channel distribution.  This is done offline, and then in the online phase we inject these channel samples into the estimator. This yields the following optimization problem
\begin{equation}\label{GAN_opt_prb}
    \mathbf{\hat{h}_{\rm GAN}} = \underset{\mathbf{h}}{\text{arg\ min\ }} ||\mathbf{y} - \sqrt{\rho}\mathbf{A}\mathbf{h}||_2^2 + r(\mathbf{h})
\end{equation}
where
\begin{align} \label{GAN_prior}
    r(\mathbf{h}) = 
    \begin{cases}
        0, & \text{if $\mathbf{h}$ is producible by the GAN}\\
       \infty, & \text{o.w.}
    \end{cases} 
\end{align}
Note that \eqref{GAN_prior} injects the a priori knowledge due to the trained GAN into the estimator, which means that among the many possible candidates, the estimate is the one that can be produced by the GAN\footnote{We note that producible means the channel is either on the range space of the generator or close to it, with the meaning of ``close" quantified in \eqref{reconst_error_high_quant}.}.

\subsection{Channel Estimation and Image Reconstruction Differences for GANs}
GANs have already been used to solve inverse problems in image processing  \cite{bora2017compressed}.  This raises the natural question of what is novel about using a GAN for channel estimation. The first answer is that the structure of the measurement matrix $\mathbf{A}$ in \eqref{chn_est_prob} is very different for the two applications. Furthermore, the signal structures of natural images and channels are distinct. To illustrate, in Fig. \ref{fig:comp_data} we visualize a sample image from the CelebA dataset by cropping $64\times16$ portion of it from the center and compare it with (i) a channel realization from a generic geometric channel model to show the spatial correlations and (ii) a channel realization from TDL-E channel model to show the frequency and time domain correlations. For ease of exposition, all signals are shaped to a $64\times 16$ matrix.  Another difference is that a good performance metric for channel estimation is the Euclidean distance or SNR, due to Gaussian noise and the dispassionate nature of symbol demodulation in the presence of such noise.  On the other hand, image quality is perceptual and qualitative, and Euclidean distance and SNR are known to be poor measures of image quality \cite{Doersch16}, \cite{Chen2020Bovik}.  Indeed, a major feature of a GAN is that it can produce an image that is far from the target image under a quantitative measure like Euclidean distance, but very close in a perceptual sense. 
\begin{figure*}[!t]
\centering
\subfigure[CelebA]{
\label{fig:celebA}
\includegraphics[width=3in]{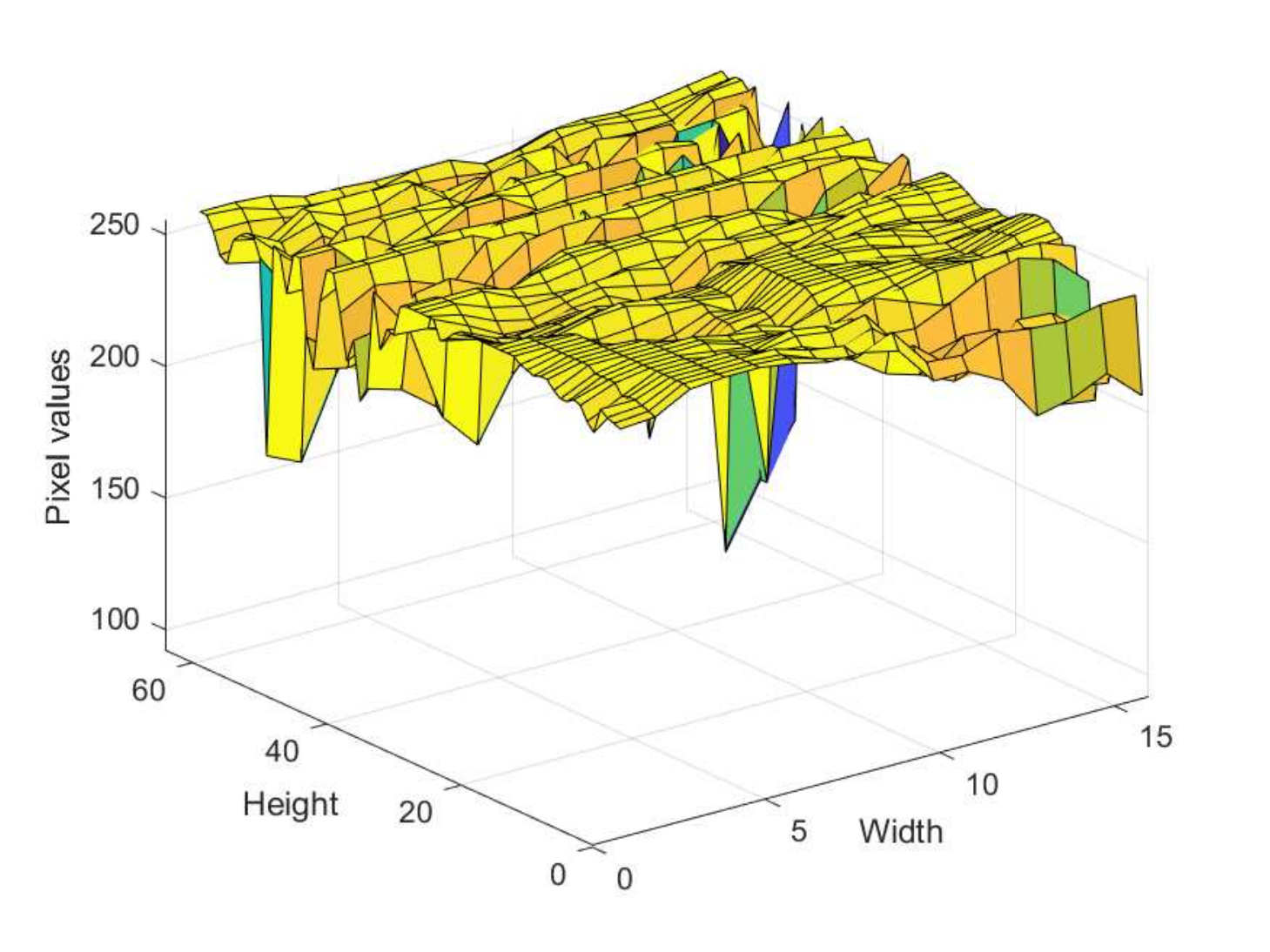}}
\qquad
\subfigure[Spatial correlations]{
\label{fig:geo_model}
\includegraphics[width=3in]{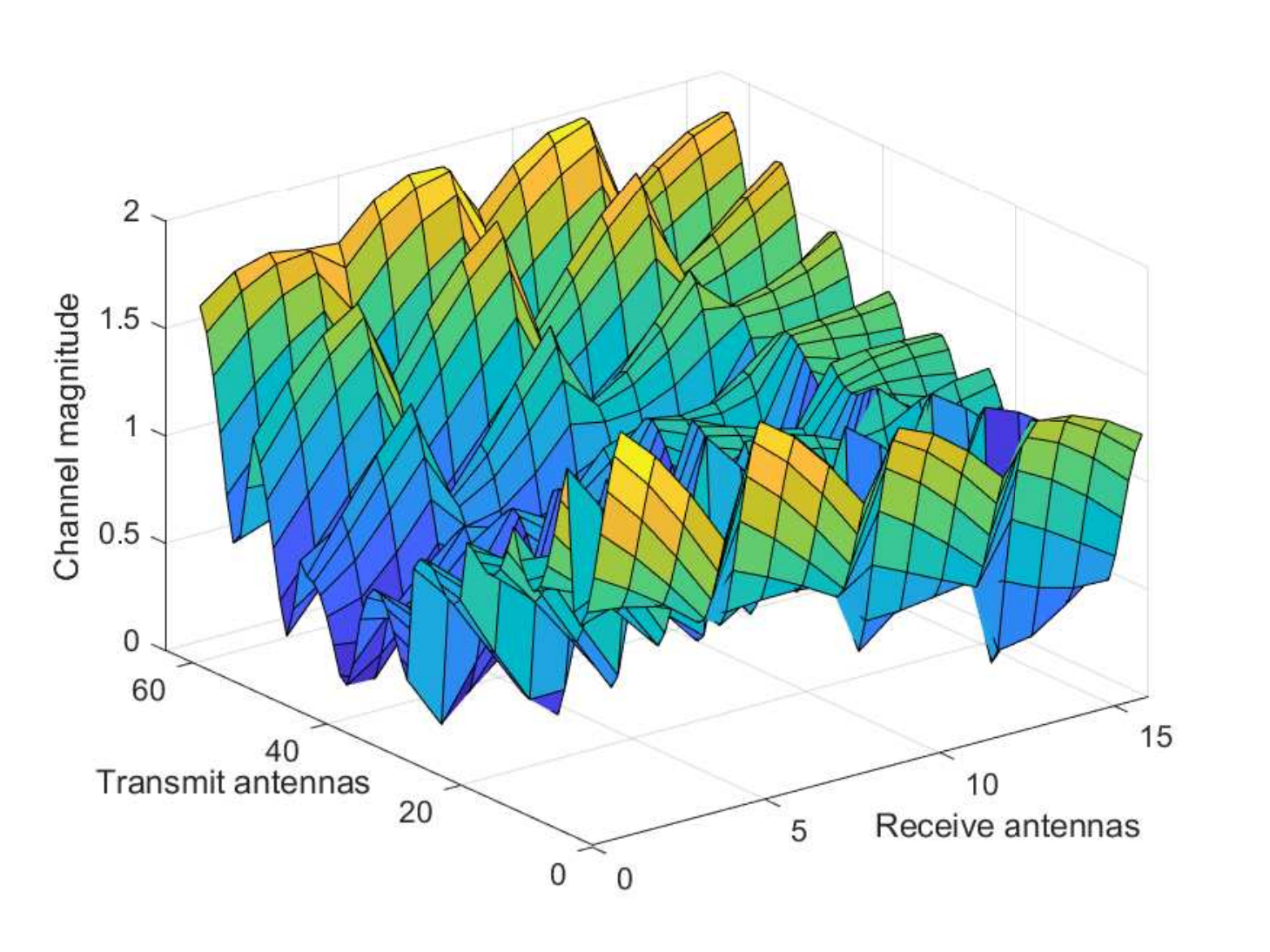}}
\subfigure[Frequency-time correlations]{
\label{fig:tdl_e_model}
\includegraphics[width=3in]{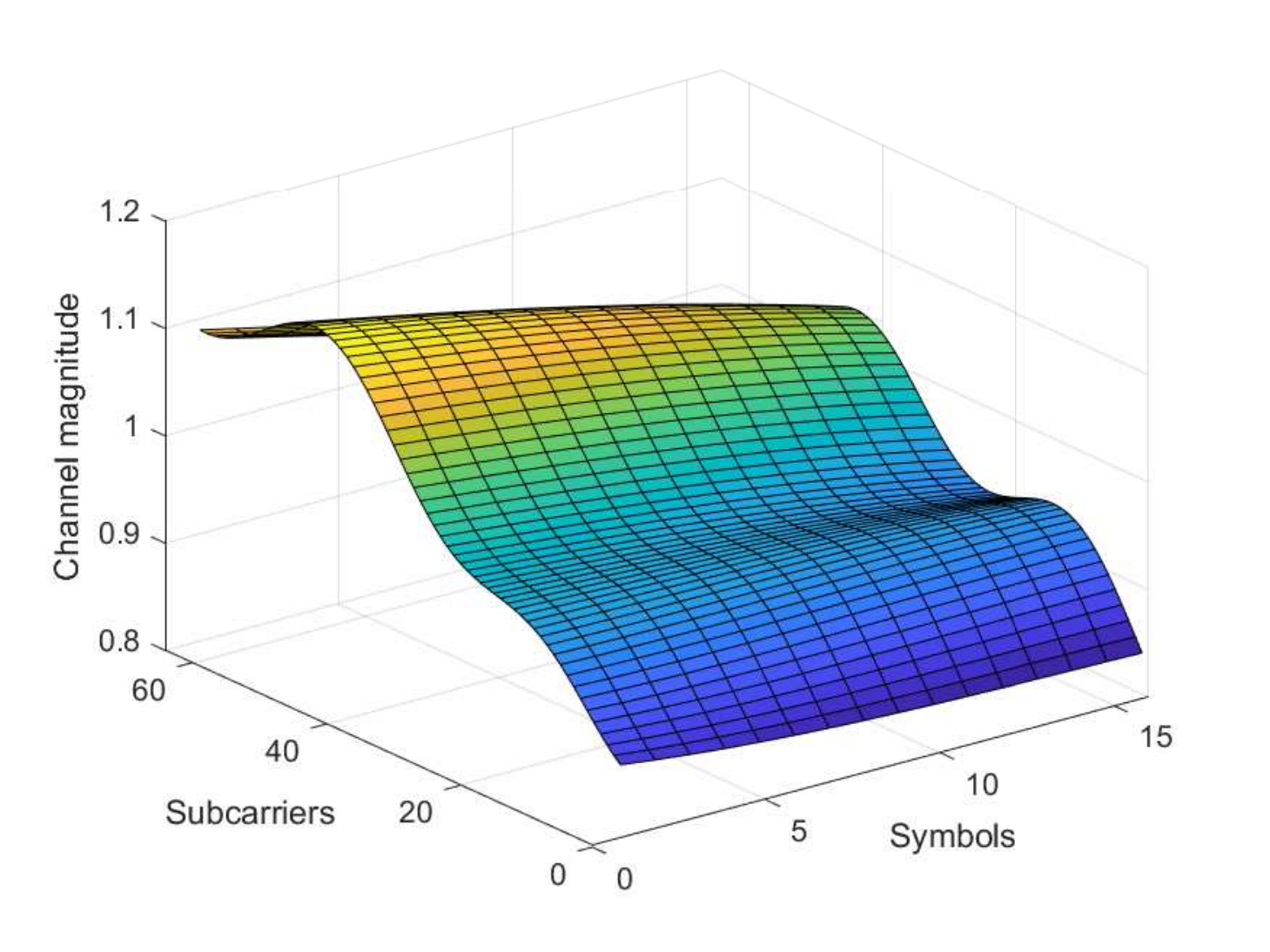}}
\caption{Visualisation of (a) a $64\times 16$ cropped image in CelebA dataset, (b) a channel realization for 64 transmit and 16 receive antennas from geometric channel model, and (c) a channel realization for 64 subcarriers and 16 OFDM symbols from TDL-E model}
\label{fig:comp_data}
\end{figure*}


\section{Channel estimation with Generative Adversarial Networks}
\label{sec:channel_estiamtion_with_training}
In this part, we present the GAN-based estimator for frequency-selective channels. Before going into the details of how the GAN works for channel estimation, the basics of GANs and our architecture are briefly summarized. Then, we explain how to solve the optimization problem in \eqref{GAN_opt_prb} at an algorithmic level.

\subsection{GAN Architecture}\label{GAN Architecture}
A GAN is composed of a generator $G_{\hat{\theta}_g}:\mathcal{R}^d\rightarrow\mathcal{R}^n$, in which $d \ll n$ and a discriminator $D_{\hat{\theta}_d}:\mathcal{R}^n\rightarrow\{-1,+1\}$, where $\hat{\theta}_g$ and $\hat{\theta}_d$ are the parameters of the generator and discriminator neural networks. As explained in \cite{goodfellow2014generative}, the discriminator is first trained both with the true samples in the dataset that are labeled as valid and the fake samples produced by the generator that are labeled as fake. In what follows, the generator is trained to enhance the quality of fake samples to fool the discriminator so that the fake samples are classified as valid. Model selection for the generator and discriminator networks is of key importance to facilitate the training. After extensive model exploration, \cite{Radford2016Chintala} developed a class of CNN architecture called \textit{Deep Convolutional GAN (DCGAN)}, and empirically showed that using a DCGAN for the generator and discriminator significantly alleviates the instability problems in training. 

In our DCGAN model, the first layer of the generator network processes the low dimensional vector with a fully connected layer and ReLU activation function, and then reshapes it into a 3-dimensional vector. This layer is then followed by four hidden layers, each of which is composed of upsampling, 2-dimensional convolution with $4\times4$ filters and $1\times1$ stride, batch normalization and ReLU activation function. It is worth noting that upsampling repeats the rows and columns so as to have the same dimensions with the channel matrix at the generator output. The output layer only involves a 2-dimensional convolution with a linear activation function. The discriminator network of our DCGAN has a 2-dimensional convolution with $3\times3$ filters and $2\times2$ stride for the input that has leaky ReLU activation function and dropout with $0.25$. This is followed by three hidden layers, each of which has 2-dimensional convolution with $3\times3$ filters and $2\times2$ stride, batch normalization and leaky ReLU activation function, and dropout with $0.25$. The slope of the leak is 0.2 for all leaky ReLU functions. The output has a fully connected layer, and its dimension is 1 to classify the channel realizations as either valid or fake. 

Consecutively repeating the process of training $D_{\hat{\theta}_d}(\cdot)$ $k$ times and training $G_{\hat{\theta}_g}(\cdot)$ one time to keep the discriminator near its optimal solution theoretically achieves the global optimum point of 
\begin{equation} \label{GAN_training_goodfellow}
	\underset{G}{\text{min\ }}\underset{D}{\text{max\ }}  \mathbf{E}_{\mathbf{h}} [\log D_{\hat{\theta}_d}(\mathbf{h})] +  \mathbf{E}_\mathbf{z}[\log (1-D_{\hat{\theta}_d}(G_{\hat{\theta}_g}(\mathbf{z})))],
\end{equation}
where $\mathbf{h}$ represents the training channel samples and $\mathbf{z}$ is the generator input sampled from a distribution $\mathcal{P(\mathbf{z})}$. The loss function in \eqref{GAN_training_goodfellow} -- which corresponds to minimizing Jensen–Shannon (JS) divergence between the generator distribution and the empirical distribution of the training channel samples -- has been used since the seminal paper \cite{goodfellow2014generative}. However, it is practically delicate to train a GAN with \eqref{GAN_training_goodfellow}, since the optimum value of $k$ is not known, and the theoretical guarantee is valid only for infinite capacity neural networks. To avoid these issues, based on the idea of replacing JS divergence with Wasserstein distance in training, \cite{arjovsky2017wasserstein} proposes to replace \eqref{GAN_training_goodfellow} with 
\begin{equation} \label{WGAN_training}
	\underset{G}{\text{min\ }}\underset{D}{\text{max\ }}  \mathbf{E}_{\mathbf{h}} [ D_{\hat{\theta}_d}(\mathbf{h})] -  \mathbf{E}_\mathbf{z}[D_{\hat{\theta}_d}(G_{\hat{\theta}_g}(\mathbf{z}))].
\end{equation}
The GANs whose generator and discriminator are trained according to \eqref{WGAN_training} is known as a \emph{Wasserstein GAN}. Note that for a Wasserstein GAN, the Lipschitz condition must be satisfied, and we ensured this with weight clipping.  We summarize the overall process of training a Wasserstein GAN according to \eqref{WGAN_training} in Algorithm \ref{algorithmGAN}, pointing the interested readers to \cite{arjovsky2017wasserstein} for more details.
\begin{algorithm}
 \caption{Wasserstein GAN for generating channels}\label{algorithmGAN}
 \begin{algorithmic}[1]
	\WHILE {$\hat{\theta}_g$ and $\hat{\theta}_d$ have not converged}
    \FOR {$t=1,2,\cdots,k$}
    \STATE Sample $\{\mathbf{h}_i\}_{i=1}^m$, a batch from the real channel realizations 
    \STATE Sample $\{\mathbf{z}_i\}_{i=1}^m \sim \mathcal{P(\mathbf{z})}$, a batch from the input prior to generate channels
    \STATE Calculate $\nabla{\hat{\theta}_d}\leftarrow \nabla_{\hat{\theta}_d} \frac{1}{m}\sum_{i=1}^m \bigg[D_{\hat{\theta}_d}(\mathbf{h_i}) -  D_{\hat{\theta}_d}(G_{\hat{\theta}_g}(\mathbf{z_i}))\bigg]$
    \STATE Update $\hat{\theta}_d$ with any gradient ascent-based method, e.g., $\hat{\theta}_d \leftarrow \hat{\theta}_d + \alpha\nabla{\hat{\theta}_d}$
    \STATE $\hat{\theta}_d \leftarrow \text{clip}({\hat{\theta}_d},-\epsilon_{\rm clip},\epsilon_{\rm clip})$
    \ENDFOR
    \STATE Sample $\{\mathbf{z}_i\}_{i=1}^m \sim  \mathcal{P(\mathbf{z})}$, a batch from the input prior to generate channels
    \STATE Calculate $\nabla{\hat{\theta}_g}\leftarrow {-\nabla_{\hat{\theta}_g}\frac{1}{m}\sum_{i=1}^mD_{\hat{\theta}_d}(G_{\hat{\theta}_g}(\mathbf{z}_i))}$
    \STATE Update $\hat{\theta}_g$ with any gradient descent-based method, e.g., $\hat{\theta}_g \leftarrow \hat{\theta}_g - \alpha\nabla{\hat{\theta}_g}$
    \ENDWHILE
 \end{algorithmic} 
 \end{algorithm}
\subsection{Frequency Selective Channel Estimator}
To solve the channel estimation problem in \eqref{GAN_opt_prb}, we first train a Wasserstein GAN offline, then extract its generator network $G_{\hat{\theta}_g}$ and iteratively optimize the input of this generative network online. This entire process is presented in Fig. \ref{fig:proposed_framework_chn_est} and detailed next.
\begin{figure}[!t]
\centering
\includegraphics[width=3.5in]{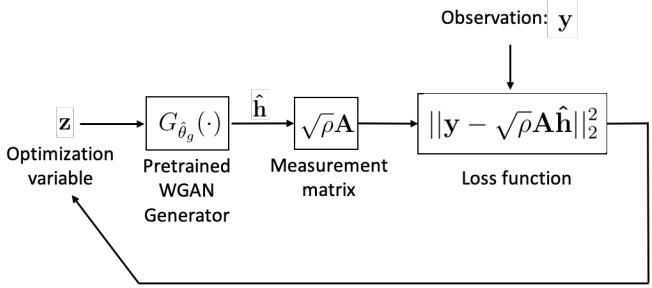}
\caption{Channel estimation by initializing the input $\mathbf{z}$ of the generative network, which is extracted from the pretrained GAN, with some random vector, e.g., with a realization of a Gaussian vector and optimizing $\mathbf{z}$ in light of the observation $\mathbf{y}$.}
\label{fig:proposed_framework_chn_est}
\end{figure}

\noindent
\textbf{1) Offline phase}: We envision that there will be datasets in future standards for different environments e.g., for urban macro (UMa), urban micro (UMi), indoor-open office just like the existing empirical channel models, and the neural network parameters are trained with these datasets. However, since there is no such dataset yet, in this paper the parameters of the GAN are trained offline by generating samples from a channel model as explained in Algorithm \ref{algorithmGAN}. After training, the generative part is taken and used as in Fig. \ref{fig:proposed_framework_chn_est}. We note that the parameters of the generative network are never retrained until there are significant changes in the channel statistics. The quantification of ``significant changes'' is a topic for future research, but we will provide an analysis and simulation results to better understand this.

\noindent \textbf{2) Online phase}: After offline training, the input of the generator network is periodically optimized according to the received signal in \eqref{chn_est_prob} to minimize
\begin{equation}\label{optimization_problem_GAN}
        \mathbf{z^*} = \underset{\mathbf{z}}{\text{arg\ min\ }} 
        ||\mathbf{y} - \sqrt{\rho}\mathbf{A}G_{\hat{\theta}_g}(\mathbf{z})||_2^2
\end{equation}
once per coherence time interval. The loss function in \eqref{optimization_problem_GAN} is different than \cite{Bal2020DosJalDimAnd}, which adds the $l_2$ norm square of $\mathbf{z}$ to \eqref{optimization_problem_GAN}, multiplied with an hyperparameter, as a regularization term. 
We note that $\mathbf{z}$ is initialized with Gaussian i.i.d. random variables and the output of the generative network for the optimized input corresponds to the channel estimate, i.e.,
\begin{equation} \label{GAN_chn_est}
        \mathbf{\hat{h}}_{\rm GAN} = G_{\hat{\theta}_g}(\mathbf{z^*}).
\end{equation}
There are many methods to solve \eqref{optimization_problem_GAN}, e.g., standard gradient descent or its variants. These steps are summarized in Algorithm \ref{algorithmCE}.
\begin{algorithm}
 \caption{GAN-based channel estimation}\label{algorithmCE}
 \begin{algorithmic}[1]
\renewcommand{\algorithmicrequire}{\textbf{Input: }}
\renewcommand{\algorithmicensure}{\textbf{Output: }}
\REQUIRE Gaussian i.i.d. noise $\mathbf{z}$
\ENSURE $\mathbf{\hat{h}}_{\rm GAN}$
\\ \textit{Offline Phase}: 
	\STATE Train the Wasserstein GAN as explained in Algorithm \ref{algorithmGAN} 
	\STATE Extract the trained generator $G_{\hat{\theta}_g}$ 
\\ \textit{Online Phase}: 
    \FOR {each coherence time interval}
    \STATE Given the noisy received signal $\mathbf{y}$ solve \eqref{optimization_problem_GAN} 
    \STATE Obtain the channel estimate as in \eqref{GAN_chn_est}
    \ENDFOR
 \end{algorithmic} 
 \end{algorithm}

\subsection{Measurement Matrix} \label{Measurement Matrix}
It was proven in \cite{bora2017compressed} that the framework in Fig. \ref{fig:proposed_framework_chn_est} reconstructs the signal with some bounded error when the measurement matrix has a sub-Gaussian distribution\footnote{A random variable $x$ has a sub-Gaussian distribution if its tail decays at least as fast as the tails of a Gaussian, or more formally $\mathcal{P}(|x|>t)\leq Ce^{-vt^2}$ for every $t > 0$, and positive constants $C$ and $v$.}. More specifically, when a ReLU generative network is utilized for $G_{\hat{\theta}_g}$, the channel estimation error is bounded by \cite{bora2017compressed} 
\begin{equation}\label{reconst_error_high_quant}
    ||\mathbf{\hat{h}}_{\rm GAN}-\mathbf{h}||_2 \leq 6\underset{ z^*\in \textbf{R}^k}{\text{min}}||G_{\hat{\theta}_g}(\textbf{z}^*)-\mathbf{h}||_2 + 3||\mathbf{w}||_2 + 2\epsilon
\end{equation}
with probability $1-e^{-\Omega(\Upsilon)}$ where $\Upsilon=N_fN_r^{\rm RF}N_p$. The first term in the right-hand side (RHS) of \eqref{reconst_error_high_quant} corresponds to the Euclidean distance between the current channel sample and the closest channel that can be produced by the trained generator, which is called representation error. The second term is the channel noise. The last term $\epsilon$ comes from gradient descent not necessarily converging to the global optimum and we denote it as optimization error. Hence, our GAN-based framework estimates the channel at worst with the given bound in \eqref{reconst_error_high_quant} when $\mathbf{A}$ is sub-Gaussian. Unfortunately, we cannot easily say that the measurement matrix $\mathbf{A}$ has a sub-Gaussian distribution, because $\mathbf{A}$ is composed of the pilots, and digital and analog precoders/combiners with some constraints. Specifically, the analog precoder and combiner matrices must satisfy the constant modulus constraint, i.e., the magnitude square of each of its element must have the same constant value. Furthermore, there is a total transmission power constraint. We now prove they do for a general class of pilots. 
\begin{theorem} \label{theorem_1}
If the pilot symbols are zero mean bounded i.i.d. random variables, then the measurement matrix $\mathbf{A}$ in \eqref{chn_est_prob} has sub-Gaussian entries for a given total transmission power.
\end{theorem}
\begin{proof}
See Appendix \ref{Proof of Theorem 1}.
\end{proof}
Thus, our channel estimator algorithm for hybrid transceivers provides a theoretical guarantee that the estimation error cannot be worse than the sum of the three error terms in \eqref{reconst_error_high_quant}.  The first term in \eqref{reconst_error_high_quant} does not diminish with high SNR although the noise power or $||\mathbf{w}||_2^2$ goes to $0$ when SNR goes to infinity. This leads to residual error even if $\epsilon$ is negligible, which becomes more prominent for high SNRs. On the other hand, if there is a perfectly trained generative network, which can produce all the samples in a channel distribution, then this residual error becomes 0.

\section{Generalization Capability}\label{Generalization Capability}
The training channel samples used for the GAN and the test channels to measure the performance of the proposed estimator must ideally have the same structure.  However, the structures in wireless channels can constantly change depending on various factors. This implies that the test channels can have different structures, and the ability to gracefully handle novel data or the generalization capability of our estimator has to be investigated. This, however, is an involved subject. In particular, determining the channel conditions that require retraining is a complicated problem, since this also depends on the GAN architecture, i.e., some architectures can be more robust to the distributional changes\footnote{To avoid online retraining, receivers can also store different set of GAN parameters after identifying the scenarios that require retraining, and pick the most appropriate set of parameters.}. Thus, we start by studying the relatively simple case, namely the number of clusters and rays, which are a linear function of the channel. 

The proposed method in Fig. \ref{fig:proposed_framework_chn_est} searches the best estimate in the channel manifold acquired with the trained generative network of the GAN, i.e.,
\begin{equation}\label{GD}
\mathbf{z}_{n+1} = \mathbf{z}_{n} - \mu_n\nabla_{\mathbf{z}_{n}}||\mathbf{y} - \sqrt{\rho}\mathbf{A}G_{\hat{\theta}_g}(\mathbf{z}_{n})||_2^2,
\end{equation}
where $\mu_n$ is the step size at the $n^{th}$ iteration. For the ease of analysis, we assume that the parameters $\hat{\theta}_g$ are trained offline with the samples that come from 
\begin{equation}\label{geo_chn_mod}
    \mathbf{H}_k = \sqrt{\frac{N_rN_t}{\gamma}}\sum_{i=1}^{N_{\rm cl}^{\rm (off)}}\sum_{l=1}^{N_{\rm ray}^{\rm (off)}}\alpha_{i,l}\mathbf{a}_r(\phi_{i,l}^r,\theta_{i,l}^r)\mathbf{a}_t(\phi_{i,l}^t,\theta_{i,l}^t)^H
\end{equation}
for $k=0,1,\cdots,N_f-1$ where $N_{\rm cl}^{\rm (off)}$ is the number of clusters, each of which has $N_{\rm ray}^{\rm (off)}$ paths. Here, $\alpha_{i,l}$ is the complex path gain of the $l^{th}$ ray in the $i^{th}$ cluster, and $\phi_{i,l}^r, \theta_{i,l}^r, \phi_{i,l}^t, \theta_{i,l}^t$ are the azimuth and elevation angles of arrival and departure respectively. For notation simplicity, we assume that the impact of pulse shape $\sum_{b=0}^{L_c-1}p(bT_s-\tau_{i,l})e^{-\frac{j2\pi kb}{N_f}}$, where $T_s$ is the symbol period and $\tau_{i,l}$ is the corresponding delay, is included within $\alpha_{i,l}$. The vectors $\mathbf{a}_r(\phi_{i,l}^r,\theta_{i,l}^r)$ and $\mathbf{a}_t(\phi_{i,l}^t,\theta_{i,l}^t)$ are the normalized receive and transmit antenna array response, which cover the relative angle of arrival and departure shift of each ray. Also, $\gamma$ is the normalization constant to ensure that $\mathbb{E}[||\mathbf{H}_k||_F^2] = N_rN_t$. 

To understand the generalization capability, we need to perceive the impact of the channel statistics on \eqref{GD}. In this regards, we assume that the online channel has a different number of clusters and rays than the offline samples, for which the GAN was trained. More precisely, we assume that the GAN is offline trained with $N_{\rm cl}^{\rm (off)}$ and $N_{\rm ray}^{\rm (off)}$ number of clusters and rays, and this yields the parameters $\hat{\theta}_g$ that define the distribution $G_{\hat{\theta}_g}$. We further assume that the GAN parameters would become $\theta_g$ leading to the distribution $G_{\theta_g}$ if it was trained with $N_{\rm cl}^{\rm (on)}$ and $N_{\rm ray}^{\rm (on)}$, the parameters of \eqref{geo_chn_mod} for the online channel. We next show that these two distributions belong to the same distribution family. 
\begin{lemma}\label{lemma1}
$G_{\hat{\theta}_g}$ and $G_{\theta_g}$ each have a sub-Gaussian distribution.
\end{lemma}
\begin{proof}
Assume that $\alpha_{i,l}$'s in \eqref{geo_chn_mod} are independent random variables, whose support set is between $[-B_{i,l},B_{i,l}]$ and the GAN is perfectly trained so that it learns the channel distribution. Then, using the Hoeffding's inequality \cite{UHD_book} given in \eqref{orig_hoeffding_inequality} for \eqref{geo_chn_mod}, the tail of $G_{\hat{\theta}_g}$ is specified as
\begin{equation}\label{est_param_dist}
    \mathcal{P}(G_{\hat{\theta}_g}\geq t) \leq \exp\left(-\frac{t^2}{2 \sum_{i=1}^{N_{\rm cl}^{\rm (off)}}\sum_{l=1}^{N_{\rm ray}^{\rm (off)}}||\mathbf{a}_{i,l}||_F^2}\right)
\end{equation}
where $\mathbf{a}_{i,l}=\sqrt{\frac{N_rN_t}{\gamma}}\mathbf{a}_r(\phi_{i,l}^r,\theta_{i,l}^r)\mathbf{a}_t(\phi_{i,l}^t,\theta_{i,l}^t)^H$. Similarly, the tail of $G_{\theta_g}$ becomes
\begin{equation}\label{param_dist}
    \mathcal{P}(G_{\theta_g}\geq t) \leq \exp\left(-\frac{t^2}{2 \sum_{i=1}^{N_{\rm cl}^{\rm (on)}}\sum_{l=1}^{N_{\rm ray}^{\rm (on)}}||\mathbf{a}_{i,l}||_F^2 }\right).
\end{equation}
From \eqref{est_param_dist} and \eqref{param_dist}, $G_{\hat{\theta}_g}$ and $G_{\theta_g}$ are sub-Gaussian.
\end{proof}

Lemma \ref{lemma1} is used to show that the gradient of the generator input parameters with respect to the measurement error
\begin{equation}\label{meas_err}
    J_n(\theta_g, \hat{\theta}_g) = ||\mathbf{y} - \sqrt{\rho}\mathbf{A}G_{\hat{\theta}_g}(\mathbf{z}_{n})||_2^2.
\end{equation}
is unbiased even if the statistics of the test channel samples differ from the training samples in terms of the number of clusters and rays. To make this point clear, let $\mathbf{h}$ be in the range space of the hypothetically online trained GAN for some input vector $\mathbf{z}$. This means that the received signal in \eqref{chn_est_prob} can be written as
\begin{equation} 
\mathbf{y} = \sqrt{\rho}\mathbf{A}G_{\theta_g}(\mathbf{z}) + \mathbf{w}.
\end{equation}
Thus, differentiating \eqref{meas_err} with respect to the generative input yields
\begin{eqnarray} \label{part_der_cost}
    \nabla_{\mathbf{z}_{n}} J_n(\theta_g, \hat{\theta}_g) &=& \nabla_{\mathbf{z}_{n}}\mathbf{y}^H\mathbf{y} - \nabla_{\mathbf{z}_{n}}2\sqrt{\rho}\mathbf{y}^H\mathbf{A}G_{\hat{\theta}_g}(\mathbf{z}_{n}) + \nonumber \\
    && \nabla_{\mathbf{z}_{n}} \rho G_{\hat{\theta}_g}(\mathbf{z}_{n})^H\mathbf{A}^H\mathbf{A}G_{\hat{\theta}_g}(\mathbf{z}_{n}).
\end{eqnarray}
Since both $G_{\hat{\theta}_g}$ and $G_{\theta_g}$ are sub-Gaussian as shown in Lemma \ref{lemma1}, the RHS of \eqref{est_param_dist} and \eqref{param_dist} converge to a Dirac delta function $\delta(t)$ when $\gamma \rightarrow \infty$, and hence for $\gamma \rightarrow \infty$ we can observe that
\begin{equation}\label{cond}
    \mathbb{E}\bigg[ \nabla_{\mathbf{z}_{n}} J_n(\theta_g, \hat{\theta}_g) \bigg] = \nabla_{\mathbf{z}_{n}} ||\mathbf{y} - \sqrt{\rho}\mathbf{A}G_{\theta_g}(\mathbf{z}_{n})||_2^2.
\end{equation}
Although these are limiting results for very high SNR ($\gamma \rightarrow \infty$) and i.i.d. channel taps, our empirical results in Section \ref{Generalization Capability Num} indicate a similar behavior for practical SNRs and channel models.



The importance of \eqref{cond} is related to our empirical observations in Sect. V.C that having mismatched numbers of clusters and rays between online and offline phases does not appear to degrade the performance. This can be explained with the policy gradient concept in reinforcement learning, since we can model our sequential decision making in the $\mathbf{z}$-space in \eqref{GD} as a a stateless (one-step) reinforcement task by taking an action $G_{\hat{\theta}_g}(\mathbf{z}_{k})$ at each step and receiving the related penalty in \eqref{meas_err}. Specifically, if we consider the trained generator $G$ as our policy, and the generator input $\mathbf{z}$ as the policy parameters, then the problem in \eqref{GD} becomes equivalent to updating the policy parameters according to a scalar performance metric. From policy gradient methods, it is known that using an estimate of the true gradient performs as well as the ground truth gradient if the expected value of the gradient estimate approximates the true gradient \cite{SuttonBarto}: which is ensured for our case via \eqref{cond}. Furthermore, with more exploration we can even enhance the channel estimates, because this enables us to converge to a better local minima \cite{SuttonBarto}.  More precisely, since the GAN generator supports vector arithmetic with some error $\mathbf{e}$ \cite{Radford2016Chintala}, defining $\mathbf{z}=\mathbf{z}_{n}+\mathbf{\Delta z}_{n}$ yields
\begin{equation} \label{chn_param_theta}
\mathbf{y} = \sqrt{\rho}\mathbf{A}(G_{\theta_g}(\mathbf{z}_{n}) + G_{\theta_g}(\mathbf{\Delta z}_{n})+\mathbf{e}) + \mathbf{w}.
\end{equation}
An analogy for \eqref{chn_param_theta} is that if $G_{\theta_g}(\mathbf{z}_{n})$ is a man's face without glasses and $G_{\theta_g}(\mathbf{\Delta z}_{n})$ is glasses, then $G_{\theta_g}(\mathbf{z}_{n}+\mathbf{\Delta z}_{n})$ is a clean image of a man with glasses, whereas $G_{\theta_g}(\mathbf{z}_{n}) + G_{\theta_g}(\mathbf{\Delta z}_{n})$ is a noisy image of a man with glasses. Since the latter is also a meaningful image, there is an empirical evidence that the error $\mathbf{e}$ is bounded; otherwise it would become a nonsense image \cite{Radford2016Chintala}. Depending on this we further assume that $\mathbf{z}_{n}$ and $\mathbf{\Delta z}_{n}$ are independent. Therefore, if $\theta_g = \hat{\theta}_g$, then \eqref{part_der_cost} becomes 0 in the first local minima, since $\nabla_{\mathbf{z}_{k}} G_{\hat{\theta}_g}(\mathbf{z}_{k}) = 0$ assuming that $\nabla_{\mathbf{z}_{k}}\mathbf{e}$ is negligible. This means that the GAN input in \eqref{GD} converges to the first local minima. 
On the other hand, if $\theta_g \neq \hat{\theta}_g$, \eqref{part_der_cost} does not become 0 even if $\nabla_{\mathbf{z}_{k}} G_{\hat{\theta}_g}(\mathbf{z}_{k}) = 0$. This enables us to explore the landscape of the generator more.

\section{Numerical Results}\label{Numerical Results}
The proposed frequency selective estimator is assessed for downlink channel estimation, in which a base station that has uniform rectangular array (URA) for $64$ antennas transmits to a receiver with $16$ URA antennas over $64$ subcarriers. The spacing among antenna elements is taken $\lambda/2$ both at the transmitter and receiver unless otherwise stated. In the simulations, a single RF chain is utilized at the receiver as an example. This heavily reduces the power consumption at the expense of decreasing the number of measurements, i.e., the number of rows in $\mathbf{A}$. To be more precise, the power consumption per RF chain at mmWave frequencies is about $300$mW, which is expected to be even higher for THz communication \cite{Rangan2013Liu}. Hence, the power consumption for the fully digital receiver becomes more than $4$W, whereas it is only around $300$mW for our receiver. Since the transmitter RF chains do not affect the number of measurements, it is taken as $N_s=N_t^{RF}$, but surprisingly further reduction is possible for the proposed estimator as will be discussed. Phase shifters both at the transmitter and receiver have one-bit quantized angles as explained in Section \ref{System Model}. Before giving the performance results, we first explain the GAN training process.

\subsection{GAN Training}

To train the GAN, we generated $5000$ complex channel realizations with dimension $(64\times16\times64\times16)$ via MATLAB using the ''5D toolbox'', which corresponds to (transmit antenna $\times$ receive antenna $\times$  subcarriers $\times$ OFDM symbols) from the TDL-E channel model that supports up to 0.1 THz \cite{LTEChannelModel}. Although it is not obligatory to train the 4-dimensional array with a single GAN, e.g., there can be $N_t$ parallel GANs, each of which can be trained with the same $(16\times64\times16)$ samples, alignment of the channels in frequency, time and spatial axes for the height, width and depth of the CNNs of the generator and discriminator heavily affects the performance. This is because a high correlation is needed for the height and width due to the upsampling in the generator. Since the correlations in spatial domain becomes relatively less than the other axes when the antenna spacing is $\lambda/2$ for moderate delay and Doppler spread, the frequency, time and space axes of the channels are aligned as height, width and depth, respectively. Furthermore, the real and imaginary parts of the channels are split and stacked in the spatial axis. 

The parameters are optimized with the training samples via RMSprop optimizer, a variant of gradient descent, with learning rate $0.00005$ for $3000$ epochs and a batch size of $200$. To understand the impacts of this setting, we compare the channel estimation error for different GAN configurations by generating $100$ test channels. Unlike the traditional neural networks that are not trained with adversary loss, it is not clear how to measure the training and validation error of a GAN, since the quality of training is perceptual. That's why,
we choose the channel estimation error as a performance metric. Furthermore, throughout the simulations the generator input dimension is set to a $15\times1$ vector for $\mathbf{z}$. As presented in Fig. \ref{fig:GAN_config}, the number of epochs and training data length have a major effect on the normalized mean square error (NMSE), whereas the batch size has relatively small impact.
\begin{figure} [!t]
\centering
\includegraphics[width=3.5in]{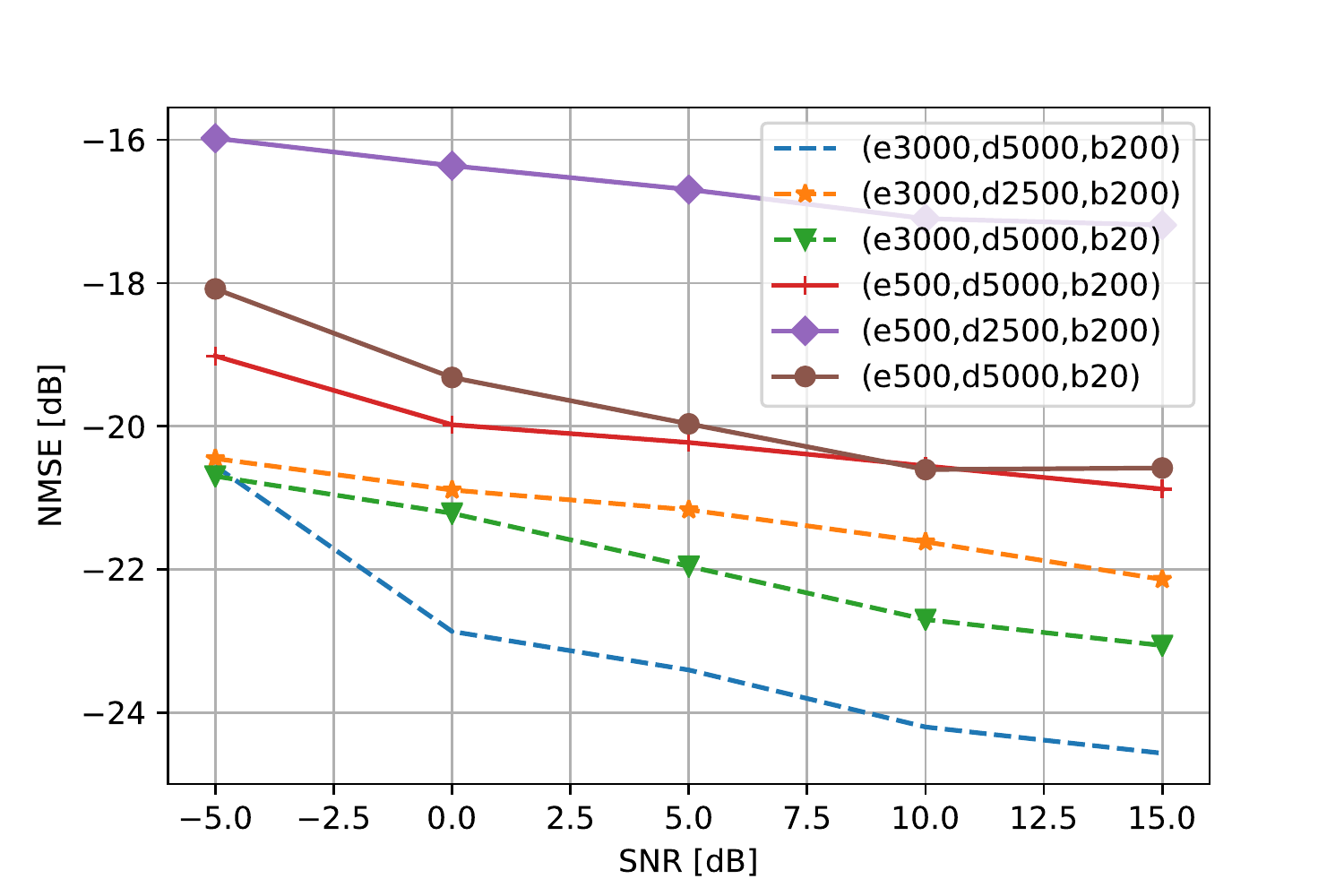}
\caption{The impact of the number of epochs, training data size and batch size while training the Wasserstein GAN on the channel estimation error, which are abbreviated as e, d, and b on the plot.}
\label{fig:GAN_config}
\end{figure}


\subsection{Performance Results}
We first assume that both offline and online channel realizations come from TDL-E model \cite{LTEChannelModel}. Furthermore, in this case orthogonal pilots are sent at each transmit antenna by placing a pilot tone in each coherence bandwidth and time. With this setting, the performance of the GAN-based channel estimation for the hybrid architecture defined above is compared with (i) the standard LS and the near-optimum LMMSE\footnote{Near-optimum LMMSE refers to the case when the noise is Gaussian, but the channel covariance matrix is estimated, i.e., not known.} channel estimators that operate on fully digital transceivers; and (ii) a supervised learning method, which maps $\mathbf{y}$ to the channel estimate $\mathbf{\hat{h}}_{\rm SL}$ through a standard ResNet model \cite{He2016Sun} after being trained with 5000 pairs of $\{\mathbf{y}, \mathbf{h}\}$.  The only differences between our approach and \cite{He2016Sun} are the linear activation function at the output layer and the input and output dimensions: our input is the stacked real and imaginary parts of $\mathbf{y}$, and the output is $\mathbf{\hat{h}}_{\rm SL}$. The performance metric is the NMSE
\begin{equation}
    {\rm NMSE} = \mathbb{E}\bigg[\frac{||\mathbf{h}-\mathbf{\hat{h}}||_2^2}{||\mathbf{h}||_2^2}\bigg],
\end{equation}
in which the expected value is over the underlying probability distribution of the channel $\mathbf{h}$, and $\mathbf{\hat{h}}$ refers to either $\mathbf{\hat{h}}_{\rm LMMSE}$, $\mathbf{\hat{h}}_{\rm LS}$, $\mathbf{\hat{h}}_{\rm GAN}$ defined in  \eqref{LMMSE}, \eqref{LS}, \eqref{GAN_chn_est} or $\mathbf{\hat{h}}_{\rm SL}$.   

Considering the worse performance of LS and high complexity of LMMSE stemming from matrix inversion and channel covariance matrix estimation, the GAN estimator seems intriguing for high frequency channel estimation as can be seen in Fig. \ref{fig:low_delay_spread_comp} and \ref{fig:high_delay_spread_comp}. Promisingly, our estimator can tackle the negative impacts of the hybrid model and achieve a close performance with respect to fully digital transceivers. Specifically, the proposed design results in very low error at low SNRs, which is the only feasible regime for mmWave and THz channel estimation. To illustrate, we can achieve the fully digital transceiver performance of the LS estimator at 20 dB SNR and LMMSE estimator at 2.5 dB SBR with our estimator at $-5$ dB SNR when the delay spread is 10 ns. Note that although LMMSE estimator is optimum for Gaussian noise when the channel covariance matrix is known, our performance is better than LMMSE for low SNRs, because of the estimation errors for channel covariance matrix. Furthermore, the residual error term in our estimator due to the representation and optimization error, which are explained in Section \ref{Measurement Matrix} and becomes more prominent at high SNRs, results in almost flat NMSE after some SNR. When the delay spread is increased to 100 ns, the performance of our estimator slightly decreases, since this yields lower correlations and hence less structures. Although our estimator at high SNRs is not as efficient as at low SNRs, this is not a problem, since the probability of observing moderate or high SNRs is extremely low due to the high propagation losses and the lack of beamforming gain. 
\begin{figure*}[!t]
\centering
\subfigure[Delay spread = 10ns]{
\label{fig:low_delay_spread_comp}
\includegraphics[width=3.25in]{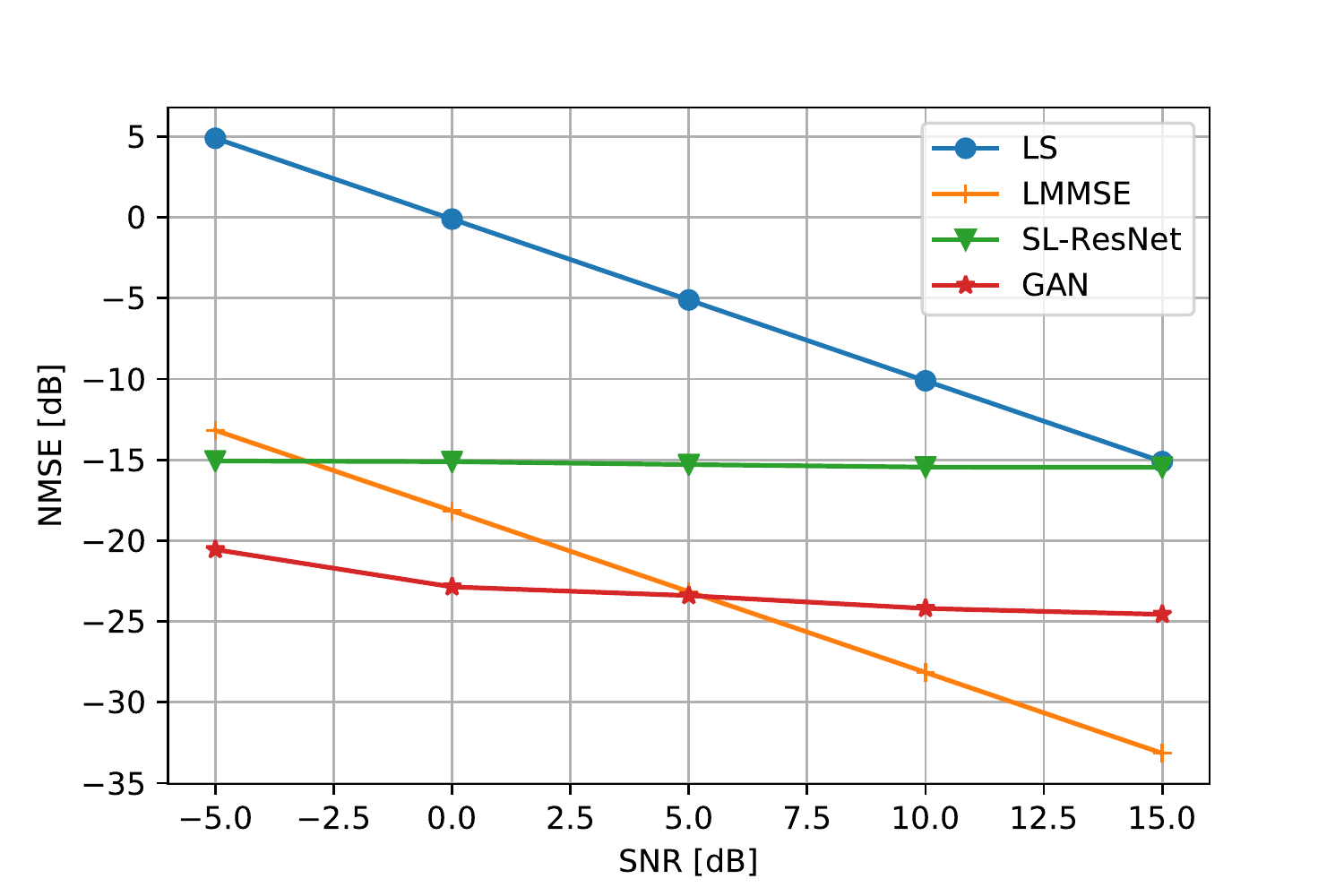}}
\qquad
\subfigure[Delay spread = 100ns]{
\label{fig:high_delay_spread_comp}
\includegraphics[width=3.25in]{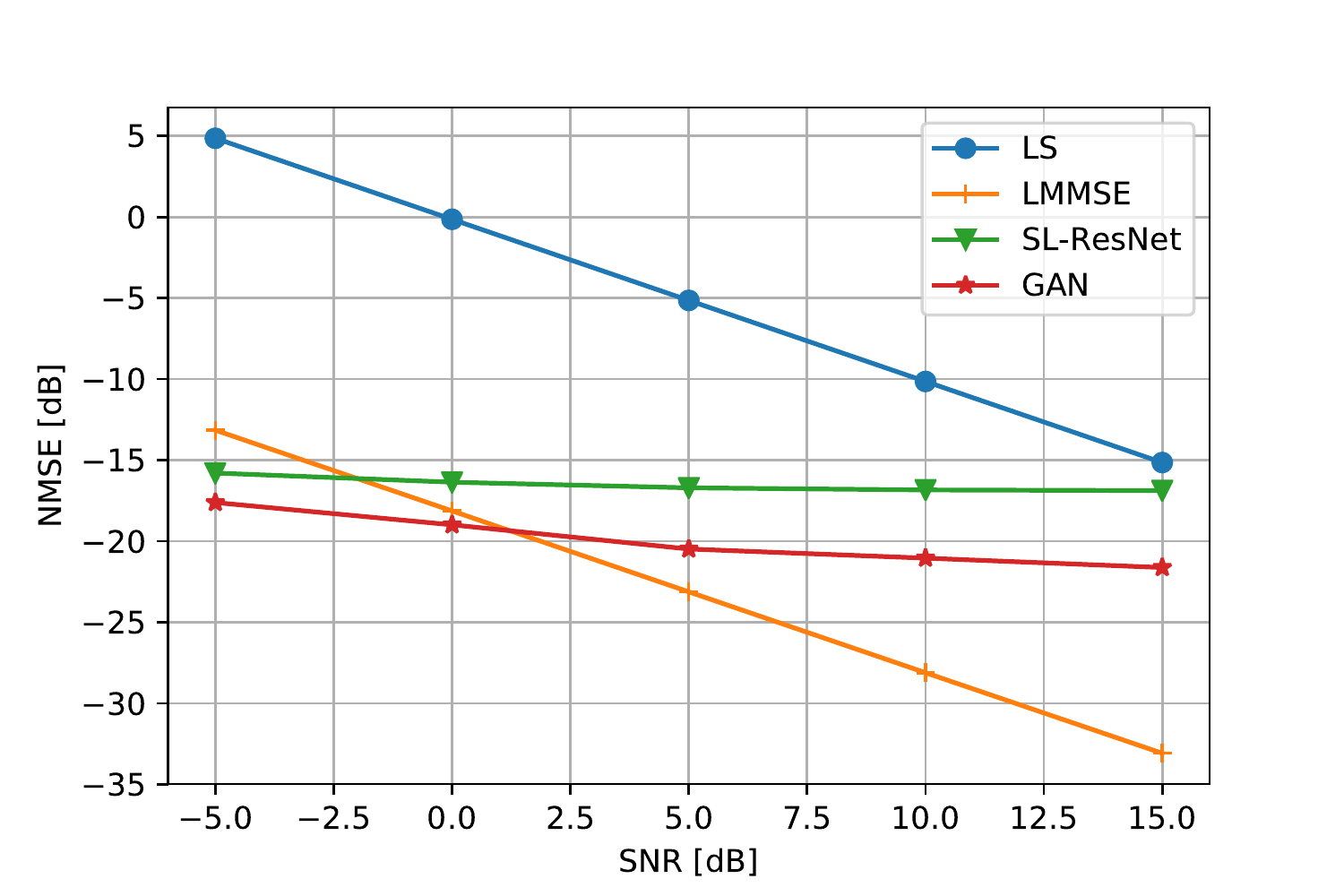}}
\caption{The comparison of the GAN-based channel estimation for hybrid architecture with the LS and LMMSE estimators for fully digital architecture, and a supervised learning method.}
\label{fig:full_pilot}
\end{figure*}

Another important point regarding Fig. \ref{fig:full_pilot} is that our GAN-based estimator outperforms the supervised learning model even if there is plenty of labeled data and a powerful CNN architecture. Specifically, to make a fair comparison, we trained the GAN and ResNet with 5000 clean channel samples via an RMSprop optimizer with a learning rate of $0.00005$ and a batch size of $200$, and use the same number of pilots. Since the ResNet architecture has many more parameters than the DCGAN architecture (33 million vs. 2 million), we trained the GAN with 3000 epochs and measure the elapsed time for channel estimation. Then, we find the number of epochs that corresponds to this amount of time for the ResNet-based channel estimator, and plot the results accordingly. We note that even if the ResNet is trained with 3000 epochs we empirically observed that its performance is still considerably worse than the GAN for low delay spread. On the other hand, for high delay spread 3000 epochs brings the ResNet-based estimator close to the GAN performance. This implies that CNNs trained with conventional loss functions are not as powerful as a GAN in exploiting the channel correlations as can be observed in Fig. \ref{fig:low_delay_spread_comp} and \ref{fig:high_delay_spread_comp}.

For mmWave and THz communication, sending orthogonal pilots for each coherence bandwidth and time at each transmit antenna leads to an excessive number of pilots due to the large number of antennas and bandwidth. In the next simulation, we address this problem by observing the estimation error of the GAN-based framework when the number of pilots is drastically reduced. Realizing that the lack of pilots creates an ill-posed problem, and thus the LS and LMMSE estimators are undefined, one has little choice but to use compressed sensing algorithms.  Thus, we use the generic OMP algorithm as a benchmark after sparsifying the channel with a DFT basis as opposed to the GAN algorithm that uses the original (unsparsified) channels. The OMP algorithm has a better performance-complexity tradeoff with respect to LASSO and does not require to know the channel distribution as opposed to message-passing algorithms \cite{Bal2020DosJalDimAnd}. However, as compared to the OMP algorithm, there is a clear performance gain of the proposed estimator as can be seen in Fig. \ref{fig:pilot_reduction}. More importantly, irrespective of the SNR the pilots can be reduced even up to $70\%$ with respect to the case of sending orthogonal pilots for each coherence bandwidth and time at each transmit antenna with only around 1 dB loss in NMSE. This brings us the flexibility of decreasing the number of pilots in frequency, time, or spatial domains or for a combination of these. Interestingly, this can also be interpreted as decreasing the number of transmit RF chains or sparsifying the connections between RF chains and antenna elements when the impact of $\mathbf{s}[n]$ is observed for $\mathbf{A}$ in \eqref{chn_est_prob}. To illustrate, decreasing the number of pilots into half can be interpreted as making $N_t^{RF}=N_s/2$.
\begin{figure} [!t]
\centering
\includegraphics[width=3.5in]{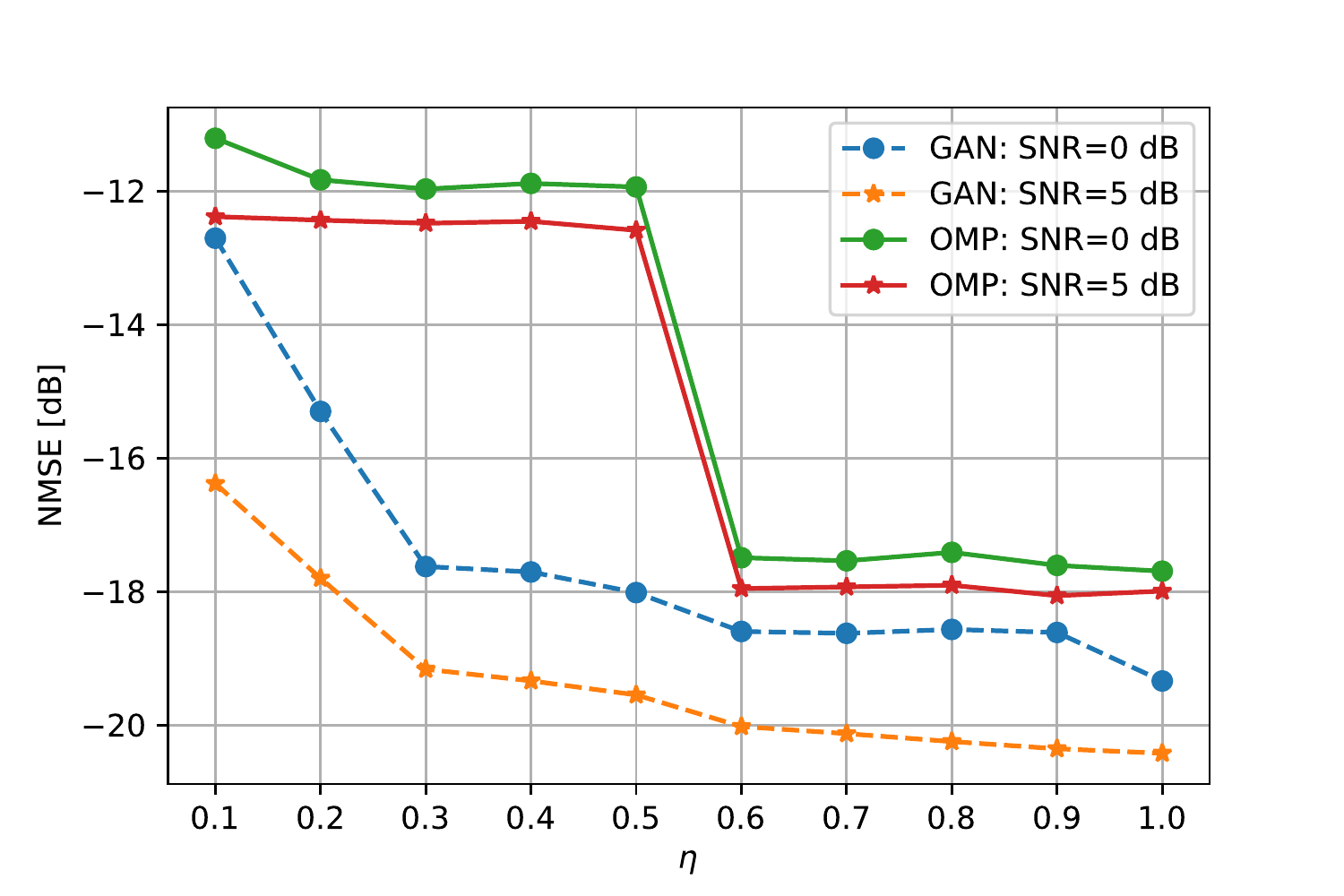}
\caption{The performance of the proposed estimator and OMP with respect to $\eta$, where $\eta$ is the ratio of the number of used pilots to the full pilot case, which corresponds to using an orthogonal pilot tone for each coherence bandwidth and time at each transmit antenna.}
\label{fig:pilot_reduction}
\end{figure}

\subsection{Generalization Capability}\label{Generalization Capability Num}
The statistics of the channel can vary with time. Thus, it is not always practical to assume that offline and online channels come from the same model. On the other hand, assuming that the offline channel comes from a model, e.g., TDL-E and the online channel from another, e.g., TDL-A makes things uncontrollable and hard to understand the impact of underlying channel parameters. As a starting point, we focus on the impact from the number of clusters and rays, and verify our analysis in Section \ref{Generalization Capability}. We consider one subcarrier, e.g., $\mathbf{H}_k$ in \eqref{chan_at_one_subcar}, and model it as a geometric channel as given in \eqref{geo_chn_mod}. The effect of the number of clusters is first observed by training the GAN with 20 clusters, each of which has 2 rays for Gaussian distributed angle of arrivals and departures with a variance of $5^\circ$. As shown in Fig. \ref{fig:clusters}, keeping all the parameters same except for the number of clusters for the online channel realizations yields that the channel estimation error does not become worse when there are 10, 15 or 25 clusters. In particular, the worst performance surfaces for 20 clusters and this is consistent with our analysis in Section \ref{Generalization Capability}.
\begin{figure} [!t]
\centering
\includegraphics[width=3.5in]{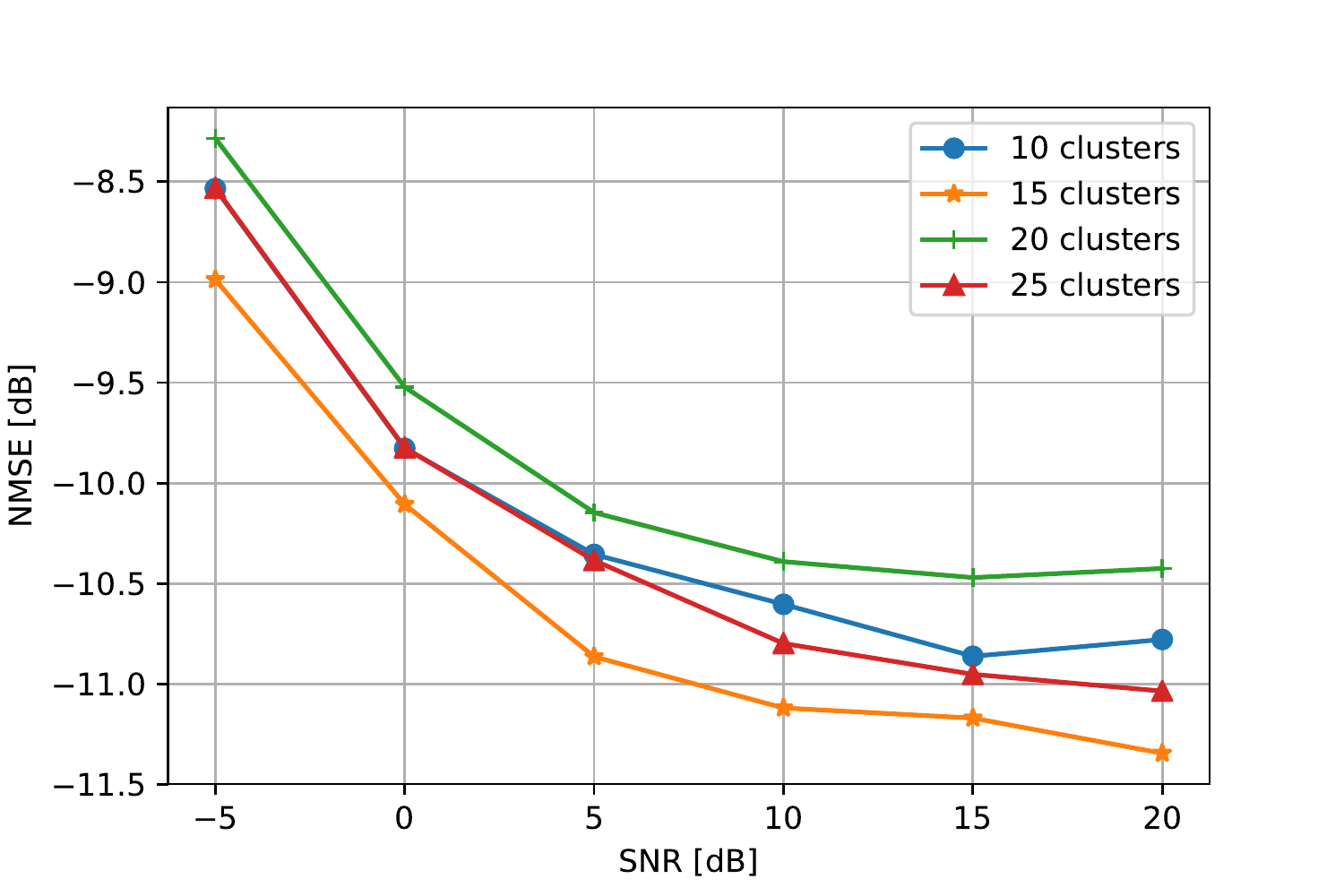}
\caption{The generalization capability of the GAN for different number of clusters when it was trained with 20 clusters}
\label{fig:clusters}
\end{figure}
Next, the same simulation is repeated for the number of rays in a cluster. In this case, we again use $20$ clusters such that each cluster has 1, 2, 4 or 8 rays. The variance of the angle spread is still $5^\circ$ for both arrival and departure beams. As can be seen in Fig. \ref{fig:rays}, although the GAN was trained with 2 rays, having 1, 4 or 8 rays per cluster does even enhance the channel estimation error similar to the previous case. Another important outcome is that although Fig. \ref{fig:clusters} and Fig. \ref{fig:rays} are obtained for one subcarrier with a very small antenna spacing of $\lambda/10$ for experimental purposes to have more structures, their NMSE are much higher in comparison to Fig. \ref{fig:full_pilot} and Fig. \ref{fig:pilot_reduction} that enjoy frequency correlations despite an antenna spacing of $\lambda/2$.
\begin{figure} [!t]
\centering
\includegraphics[width=3.5in]{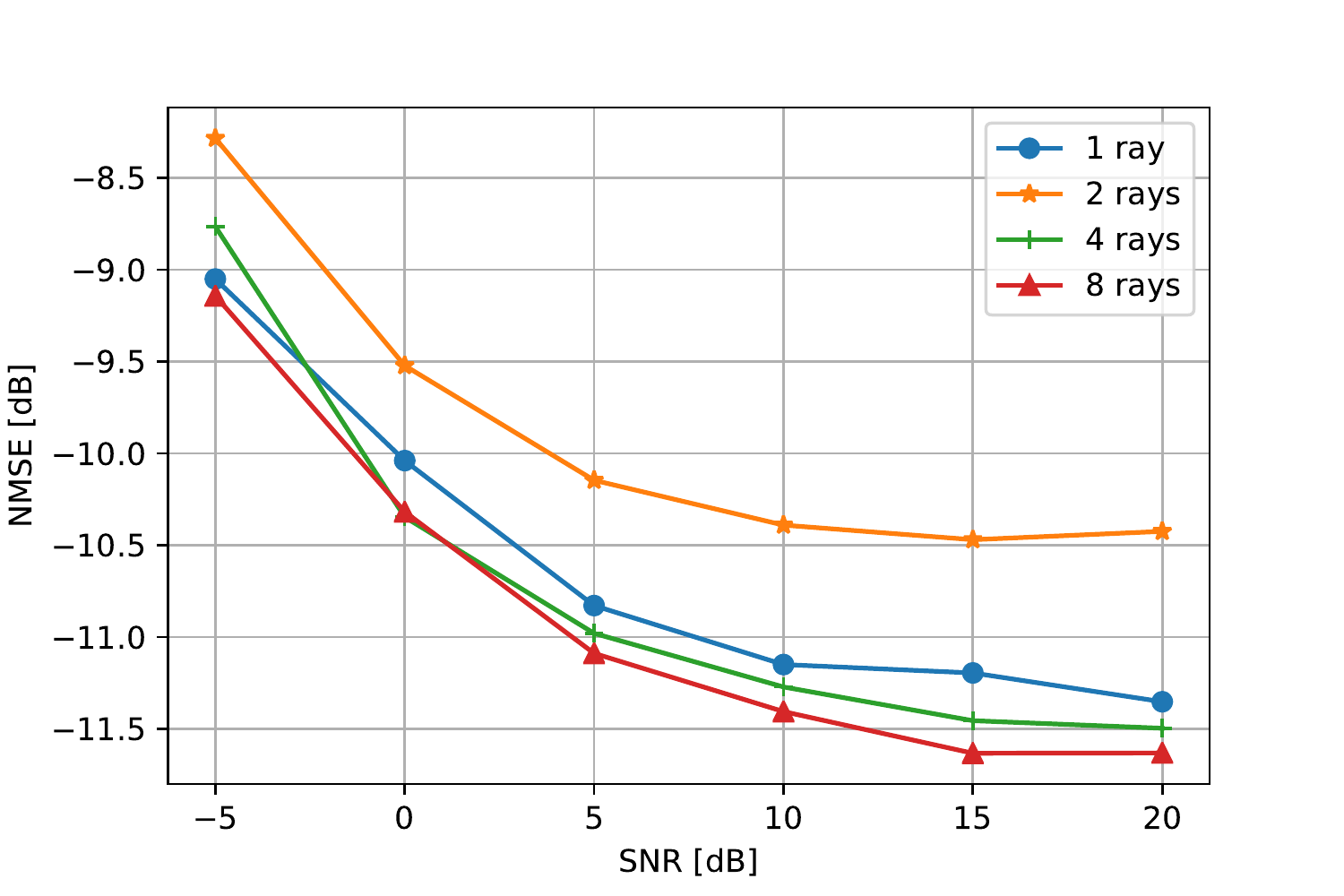}
\caption{The generalization capability of the GAN for different number of rays when it was trained with 2 rays}
\label{fig:rays}
\end{figure}

\subsection{Computational Complexity}
Relying on the generalization capability of neural networks and assuming that the channel statistics vary slowly, the training complexity of the generative parameters in the offline environment becomes negligible since it is amortized over a long time period. Thus, we focus our analysis on the inference complexity that comes from (i) optimizing the generator input $\mathbf{z}^*$ at each epoch using gradient descent as given in \eqref{GD} and (ii) estimating the channel through the matrix-vector products as in \eqref{GAN_chn_est}. For the ease of analysis, we first consider the complexity between a single transmit and receive antenna, and then generalize this for all antennas. For the first step, in the forward propagation the first layer of the generative network brings a complexity of $\mathcal{O}(N_fN_pd)$, where $d$ is the dimension of $\mathbf{z}$, and dominates the complexity of the other convolutional layers that have fewer parameters and smaller matrix multiplications. Then, to compute the gradient in \eqref{GD}, which is equal to 
\begin{eqnarray} \label{eq:grad_G}
    \nabla_{\mathbf{z}_{n}}||\mathbf{y} - \sqrt{\rho}\mathbf{A}G_{\hat{\theta}_g}(\mathbf{z}_{n})||_2^2 &=& -2\sqrt{\rho} \nabla_{\mathbf{z}_{n}} G_{\hat{\theta}_g}(\mathbf{z}_{n})^T\mathbf{A}^T  \nonumber \\
    &&(\mathbf{y} - \sqrt{\rho}\mathbf{A}G_{\hat{\theta}_g}(\mathbf{z}_{n})),   
\end{eqnarray}
the output of the generative network is multiplied with the measurement block diagonal matrix $\mathbf{A}$ that has $N_p$ blocks, and this yields a complexity of $\mathcal{O}(N_pN_f^2)$, which is the same with multiplying $\mathbf{A}^T$ with $(\mathbf{y} - \sqrt{\rho}\mathbf{A}G_{\hat{\theta}_g}(\mathbf{z}_{n}))$. The backward propagation complexity for $\nabla_{\mathbf{z}_{n}} G_{\hat{\theta}_g}(\mathbf{z}_{n})$ has nearly twice the execution time of the forward propagation in CNNs \cite{He2015Sun}, and the computational complexity of the second step only covers the forward propagation of the first step. Thus, the overall complexity becomes $\mathcal{O}(N_tN_rN_pN_f^2)$. This is much less than the LS, LMMSE and OMP estimators. Specifically, the LS and LMMSE estimators require the inversion of $\mathbf{A}^H\mathbf{A}^{-1}$ as can be seen in \eqref{LMMSE} and \eqref{LS}, and this yields a complexity of $\mathcal{O}(N_t^3N_r^3N_pN_f^3)$. The complexity of the OMP is related with the target signal dimension $\mathbf{h}$ and it is $\mathcal{O}(N_t^3N_r^3N_f^3)$ \cite{vila2013expectation}.


\section{Conclusions}
This paper demonstrates how to leverage GANs for effective frequency selective channel estimation in a mmWave or a THz channel: that is, a low SNR channel with a large number of spatial elements. Promisingly, the proposed GAN-based channel estimation works well for a hybrid architecture with one-bit quantized phase angles, and can even outperform estimators designed for fully digital receivers.  Additionally, the proposed estimator enables a substantial reduction in the number of pilot tones. Regarding its generalization capability, we demonstrate that changes in the number of clusters and rays in multipath channels can be inherently handled without retraining the generative network.  As future work, the impact of some other channel parameters such as power delay spread, Doppler spread and the angles of arrival and departure can be analyzed.  Furthermore, using a different generative model like a VAE instead of a GAN can be another extension.

\appendices
\section{Proof of Theorem 1} \label{Proof of Theorem 1}
The measurement matrix due to a single subcarrier becomes
\begin{equation} \label{meas_4_one_sub}
    \mathbf{A^{\rm (sub)}} = \mathbf{I}_{\rm N_p}\otimes(\mathbf{p}[n]^T\mathbf{F}_{\rm BB}[n]^T\mathbf{F}_{\rm RF}^T\otimes \mathbf{W}_{\rm RF}^H),
\end{equation}
in which the frequency index $k$ is dropped for simplicity. As can be observed, the statistics of $\mathbf{A}$ in \eqref{chn_est_prob} are the same as $\mathbf{A^{\rm (sub)}}$ in \eqref{meas_4_one_sub}. Thus, without any loss of generality we find the distribution of $\mathbf{A^{\rm (sub)}}$. Stacking the block diagonal matrices of \eqref{meas_4_one_sub} yields a more compact expression
\begin{equation} \label{stacked_meas_4_one_sub}
    \mathbf{\tilde{A}^{\rm (sub)}} = \mathbf{P}^T\mathbf{F}_{\rm BB}[n]^T\mathbf{F}_{\rm RF}^T\otimes \mathbf{W}_{\rm RF}^H
\end{equation}
where
\begin{equation} \nonumber
     \mathbf{P}^T = \begin{pmatrix}\mathbf{p}[n]^T \\ \vdots \\ \mathbf{p}[n+N_p-1]^T\end{pmatrix}.
\end{equation}
Since it is not practical to change $\mathbf{F}_{\rm RF}$ and $\mathbf{W}_{\rm RF}$ for each channel realization, we assume that they are given and fixed. Considering the distribution of  \eqref{stacked_meas_4_one_sub} due to an element of $\mathbf{W}_{\rm RF}^H$ results in
\begin{equation} 
    \mathbf{\tilde{A}} = w_{\rm RF}\mathbf{P}^T\mathbf{F}_{\rm BB}[n]^T\mathbf{F}_{\rm RF}^T.
\end{equation}
Notice that $ \mathcal{P}(\mathbf{\tilde{A}}) = \mathcal{P}(\mathbf{\tilde{A}^{\rm (sub)}})$, since $\mathbf{P}^T\mathbf{F}_{BB}[n]^T\mathbf{F}_{RF}^T$ is repeated for each element of $\mathbf{W}_{\rm RF}^H$.

The $(m,n)$th element of $\Tilde{\mathbf{A}}$ can then be written as 
\begin{equation}\label{a_m_n}
    \Tilde{a}_{m,n} = w_{\rm RF}\sum_{j=1}^{N_t^{RF}}\left(\sum_{i=1}^{N_s}p_{m,i}^Tf_{i,j}^T\right)\Tilde{f}_{j,n}^T
\end{equation}
where $p_{m,i}^T$ is the $(m,i)$th element of $\mathbf{P}^T$, $f_{i,j}^T$ is the $(i,j)$th element of $\mathbf{F}_{BB}[n]^T$ and $\Tilde{f}_{j,n}^T$ is the $(j,n)$th element of $\mathbf{F}_{RF}^T$. Due to the constant modulus constraint $|w_{\rm RF}|^2 = 1/N_r$ and
\begin{equation}\label{const_sum}
    \sum_{j=1}^{N_t^{RF}}\sum_{i=1}^{N_s}|f_{i,j}^T \Tilde{f}_{j,n}^T|^2 = \frac{||\mathbf{F}_{\rm BB}[n]||_F^2}{N_t}.
\end{equation}
Since there is a total transmission power constraint and Frobenius norm is submultiplicative, $||\mathbf{F}_{\rm RF}\mathbf{F}_{\rm BB}[n]||_F^2 = N_s \leq ||\mathbf{F}_{\rm RF}||_F^2||\mathbf{F}_{BB}[n]||_F^2$. Hence,
\begin{equation}
    \sum_{j=1}^{N_t^{RF}}\sum_{i=1}^{N_s}|f_{i,j}^T \Tilde{f}_{j,n}^T|^2 \geq \frac{N_s}{N_tN_t^{RF}}.
\end{equation}
According to Hoeffding's inequality, for zero mean bounded i.i.d. random variables $X_1, \cdots X_N$ and $b=(b_1, \cdots, b_N) \in \mathcal{R}^N$, we have
\begin{equation} \label{orig_hoeffding_inequality}
    \mathcal{P}\left(\sum_{i=1}^{N}X_ib_i \geq t\right) \leq \exp\left(-{\frac{t^2}{2||b||_2^2}}\right)
\end{equation}
for $t\geq0$. Replacing $X_i$ with $p_{m,i}^T$ and $b_i$ with $w_{\rm RF}\sum_{j=1}^{N_t^{RF}}f_{i,j}^T\Tilde{f}_{j,n}^T$ in \eqref{a_m_n} yields
\begin{equation} \label{hoeffding_inequality}
    \mathcal{P}\left( \Tilde{a}_{m,n} \geq t\right) \leq \exp\left(-{\frac{t^2N_tN_t^{RF}N_r}{2N_s}}\right).
\end{equation}
Thus, $\Tilde{a}_{m,n}$ has a sub-Gaussian distribution. Since \eqref{hoeffding_inequality} holds for all $m={1,2,\cdots,N_p}$ and $n={1,2,\cdots,N_t}$, this completes the proof.

\bibliographystyle{IEEEtran}
\bibliography{Balevi}

\end{document}